\documentclass[10pt]{IEEEtran}
\IEEEoverridecommandlockouts
\usepackage{fancyhdr} 
\usepackage{lastpage} 
\usepackage{extramarks} 
\usepackage{graphicx} 
\usepackage{lipsum} 
\usepackage{amssymb}
\usepackage{color}
\usepackage{amsmath}
\usepackage[final]{pdfpages}
\usepackage{paralist}
\usepackage[usenames,dvipsnames]{xcolor}
\usepackage{amsthm,amsmath,amsfonts,ed,empheq}
\usepackage{xspace}
\usepackage{url}
\usepackage{multirow}
\usepackage{footnote}
\usepackage{amssymb,amsmath,mathtools,fixmath}
\usepackage{algorithmicx,algorithm}
\usepackage{algpseudocode} 
\usepackage{cases}
\usepackage{cite}
\usepackage{graphicx,amsfonts,psfrag,url,booktabs}
\usepackage{comment}
\usepackage{cite}
\usepackage{xcolor}
\usepackage[caption=false,font=footnotesize]{subfig}
\usepackage{makecell}
\usepackage{psfrag}
\usepackage{tabularx}
\usepackage[table]{colortbl}
\usepackage{booktabs}
\usepackage{balance}
\usepackage{hhline}
\usepackage{pgf, tikz, pgfplots}
\usepackage{setspace}

\usetikzlibrary{arrows}

\input{mysymbol.sty}

\newcommand{\transp}{\mathsf{T}}

\newtheorem{theorem}{Theorem}
\newtheorem{corollaryP}{Corollary}[proposition]
\newtheorem{corollaryT}{Corollary}[theorem]
  
\newtheorem{definition}{Definition}

\begin{document}

\pagenumbering{gobble}

\title{Quantization Analysis and Robust Design for Distributed Graph Filters} 
\author{Leila Ben Saad,~\IEEEmembership{Member,~IEEE}, Baltasar Beferull-Lozano,~\IEEEmembership{Senior Member,~IEEE}, and 
Elvin Isufi,~\IEEEmembership{Member,~IEEE}
\vspace{-0.5cm}
\thanks{This work was supported by the PETROMAKS Smart-Rig grant
244205/E30, 
the TOPPFORSK WISECART grant 250910/F20 and the IKTPLUSS INDURB grant 270730/O70 from the Research Council of Norway.
Part of this work has been presented in \cite{BenSaadIsufiLozano2019}.} 
\thanks{L. Ben Saad and  B. Beferull-Lozano are with the WISENET center, Department of Information and Communication Technology, University
of Agder, Norway   (e-mails: leila.bensaad@uia.no, baltasar.beferull@uia.no). 
}
\thanks{E. Isufi is with the Faculty of Electrical Engineering, Mathematics
and Computer Science, Delft University of Technology, 2628 CD Delft,
The Netherlands (e-mail: e.isufi-1@tudelft.nl).
}
}

\maketitle
\markboth{Submitted to IEEE TRANSACTIONS ON SIGNAL PROCESSING}%
{Shell \MakeLowercase{\textit{et al.}}: Bare Demo of IEEEtran.cls for IEEE Journals}

\begin{abstract}

Distributed graph filters have recently found applications in wireless sensor networks (WSNs) to solve distributed tasks such as reaching consensus, signal denoising, and reconstruction. However, when employed over WSN, the graph filters should deal with the network limited energy, processing, and communication capabilities. Quantization plays a fundamental role to improve the latter but its effects on distributed graph filtering are little understood. WSNs are also prone to random link losses due to noise and interference. In this instance, the filter output is affected by both the quantization error and the topological randomness error, which, if it is not properly accounted in the filter design phase, may lead to an accumulated error through the filtering iterations and significantly degrade the performance. 
%
%
In this paper, we analyze 
how quantization affects distributed graph filtering over both time-invariant and time-varying graphs. We bring insights 
on the quantization effects for the two most common graph filters: the finite impulse response (FIR) and autoregressive moving average (ARMA) graph filter. 
Besides providing a comprehensive analysis, we  devise theoretical performance guarantees on the filter performance when the quantization stepsize is fixed or changes dynamically over the filtering iterations. 
For  FIR filters, we show  that a dynamic quantization stepsize leads to  more control on the quantization noise than the fixed-stepsize quantization. For ARMA graph filters, we show that 
decreasing  the quantization stepsize over the iterations reduces the quantization noise to zero at the steady-state.
In addition, we  propose   robust filter design strategies that minimize the quantization noise for both time-invariant and time-varying networks.
Numerical experiments on synthetic and two real data sets corroborate our findings and show the different trade-offs between quantization bits, filter order, and robustness to topological randomness.

\end{abstract}



\vspace{-0.28cm}
\section{Introduction}\label{sec.intro}


\IEEEPARstart{G}{raph} filters are enjoying an increasing popularity in graph signal processing (GSP) and graph convolutional neural networks \cite{Shuman2013, GamaIsufiLeusRibeiro2020}. Their ability to be convolved with a graph signal renders graph filters versatile in a variety of applications ranging from recommender systems to spectral clustering \cite{huang2016graph,BurkhardtJay2019,huang2018rating,sandryhaila2014discrete,girault2014semi,tremblay2016compressive}. 
Graph filters find also application in wireless sensor networks (WSNs) \cite{loukas2014identify,IsufiLoukasSimonettoLeus17,eisenRibeiro2019,IsufiLoukasPerraudin2019}. Here, the signal represents the sensor measurements and the WSN serves as a platform to perform distributed operations as well as a proxy to represent signal similarities in adjacent sensor nodes. Graph filters are useful for distributed signal representation \cite{NassifRichardChen2018}, reconstruction \cite{segarra2016reconstruction,IsufiLorenzoBanelliLeus2018},  denoising \cite{shuman2011distributed,ChenSandryhaila14},  consensus \cite{sandryhaila2014finite,TolstayaGamaPaulosPappas2019}, and network coding \cite{SegarraMarques17}. Motivated by these applications, this paper focuses on distributed graph filtering.


Distributed graph filtering can be implemented with two types of recursions over the nodes: finite impulse response (FIR) and autoregressive moving average (ARMA) recursions. In FIR graph filters, neighboring nodes communicate the  input signal for a finite number of iterations \cite{shuman2011distributed,sandryhaila2013discrete, SegarraMarques17,CoutinoIsufiLeus2018}. In ARMA graph filters, neighboring nodes communicate both the  input and former iterative output signal. Both implementations can be used interchangeably and often lead to a different tradeoff between accuracy and robustness to topological perturbations. The works in \cite{IsufiLoukasSimonettoLeus17,levieIsufiKutyniok2019} show that ARMA graph filters are more robust than FIRs to deterministic topological changes (e.g., sensor movements), while \cite{IsufiLoukasSimonetto2017} shows that higher order FIR graph filters suffer less random topological changes (e.g., link losses).

For either implementation, in distributed filtering over WSNs, we should account for the finite resources of individual sensors, which have limited energy, processing, and communication capabilities. To improve network efficiency in the latter aspect, quantization plays an important  role prior to data communication. 
Quantization has been extensively studied in  distributed  systems in the context of communications and signal processing    through consensus algorithms \cite{AysalCoatesRabbat2007, AysalCoates2008,FangLi2009, CarliFagnaniFrascaZampieri2010, KarMoura2010,CaiIshii2010,ThanouKokiopoulouFrossard2013, ZhuChen2016}, which  present many similarities with graph filtering from a distributed problem point-of-view. 
The importance of finite resources from graph signal processing perspective  has been recently recognized in \cite{ChamonRibeiro2017,nobreFrossard2019,ThanouFrossard2018}. In particular, \cite{ChamonRibeiro2017} --the most related to our work-- discusses the impact of fixed-stepsize quantization on FIR graph filters. 
The work in \cite{ThanouFrossard2018} approximates the graph  spectral dictionaries as polynomials of the graph Laplacian operator and learns polynomial dictionaries robust to signal quantization.
Finally, \cite{nobreFrossard2019} develops an adaptive quantization scheme for FIR graph filters that 
minimizes the quantization errors by bounding the exchanged messages and optimizing bit allocation.
While relevant contributions on the quantization aspects of graph filtering, these works focus solely on FIR graph filters.
Further, they consider time-invariant WSN topologies. This is a limitation in WSN since sensor nodes are prone to local malfunctions or communication links fall with a certain probability. 

In this work, we extend current art and evaluate quantization effects of both FIR and ARMA graph filters on time-invariant and time-varying topologies. Besides providing a broader analysis, we also devise theoretical performance guarantees on the filter performance when the quantization stepsize is fixed or changes dynamically over the filtering iterations. Further, we consider dithered quantization \cite{Schuchman1964,LipshitzWannamaker1992} to make the assumption of  quantization noise uncorrelated with input signals over the different graph filter iterations hold; an assumption commonly made in other current works but unjustified. Our analysis sheds light on different tradeoffs in distributed filtering over WSN: FIR versus ARMA graph filter; fixed-stepsize quantization versus dynamically decreasing quantization stepsizes; and quantization rate versus link loss probability. The broad research question we are interested in is \emph{how quantization affects distributed graph filtering over both time-invariant and time-varying graphs.} The specific contributions on how we answer this question are fourfold.

\begin{enumerate}
\item We study the quantization effects of distributed FIR graph filters. 
We analyze  the impact of fixed and dynamic quantization stepsize on the filtering performance and analyze their tradeoffs.
We show that a dynamic quantization stepsize allows more control on the quantization mean squared error (MSE) than fixed-size quantization.
We devise also a robust filter
design that minimizes the quantization noise. 

\item We study the quantization effects of distributed ARMA graph filters. 
We analyze the impact of
 fixed and dynamic quantization stepsize on the filtering performance and analyze their tradeoffs.
We develop an ad-hoc dynamic quantization stepsize framework that reduces the quantization MSE to zero at steady-state.

\item We conduct a statistical analysis to quantify the quantization effects on FIR and ARMA  graph filters  over random time-varying networks. 
We propose a novel filter design strategy that is robust to quantization and topological changes.

\item We characterize the different tradeoffs between the FIR and ARMA graph filters in terms of fixed-stepsize versus dynamically decreasing quantization stepsizes and between the quantization rate and the link loss probability.

\end{enumerate}

The rest of this paper is organized as follows. Section~\ref{sec:back} provides the background material. 
Sections~\ref{sec:FIRQeff} and \ref{sect-ARMA-quantization-effects}  analyze the quantization effects on FIR and ARMA graph filters, respectively. 
Section~\ref{sec.tvGraphs} contains the quantization analysis for random time-varying graphs. 
Section~\ref{sect-results} presents the numerical results. 
The paper conclusions are provided in Section~\ref{sect-conclusion}.



\vspace{-0.19cm}
\section{Background}\label{sec:back}



Consider a graph $\ccalG = (\ccalV, \ccalE)$ with node set $\ccalV = \{1, \ldots, N\}$ and edge set $\ccalE$ composed of tuples $(j,i)$ if there is a link between nodes $j$ and $i$. The set of all nodes connected to node $i$ is denoted by $\ccalN_i = \{j \in \ccalV| (j,i) \in \ccalE\}$. The graph is represented by its adjacency matrix $\bbA$ whose $(j,i)$th entry is nonzero only if nodes $j$ and $i$ are connected. If the graph is undirected, it can also be represented by the graph Laplacian matrix $\bbL$. To keep the discussion general  for both directed and undirected graphs, we will use the so-called graph shift operator matrix $\bbS$, which has plausible candidates $\bbA$, $\bbL$ or any of their normalized and translated forms. We shall only assume that the shift operator has an upper bounded spectral norm, i.e., $\|\bbS\|_2 \le \rho$.

On the vertices of $\ccalG$, a graph signal can be defined as a map from the vertex set (node set)  to the set of real numbers, i.e., $x: \ccalV \to \reals$. We can denote the graph signal by a vector $\bbx = [x_1, \ldots, x_N]^\transp$, whose $i$th entry $x_i$ denotes the signal at node $i$. WSNs match the above terminology: the nodes represent the sensors; the edges the communication links; and the signal the sensor data. 
By considering the eigendecomposition of the graph shift operator 
$\bbS = \bbU\bbLambda\bbU^{-1}$ with eigenvector matrix $\bbU = [\bbu_1, \ldots, \bbu_N]$ and diagonal eigenvalue matrix $\bbLambda = \diag(\lambda_1, \ldots, \lambda_N)$, 
we can alternatively analyze the graph signal $\bbx$ by projecting it onto the shift operator eigenspace as $\hbx = \bbU^{-1}\bbx$. This projection is referred to as the graph Fourier transform (GFT) because the $i$th element $\hat{x}_i$ denotes how much eigenvector $\bbu_i$ represents the variation of $\bbx$ over $\ccalG$ and because the variation of the different eigenvectors can be ordered. The inverse GFT is $\bbx = \bbU \hbx$ and the eigenvalue $\lambda_i$ denotes the $i$th graph frequency \cite{Shuman2013,sandryhaila2013discrete}.

\vspace{-0.35cm}
\subsection{Graph filter}


A filtering operation on a graph  combines locally  the signal from node $i$ and the signals $\{x_j\}$ from all neighbors $j \in \ccalN_i$ of node $i$  to produce the output:
\vspace{-0.05cm}
\begin{equation}\label{eq.locFilt}
y_i = \sum_{j \in \ccalN_i \cup i}\phi_{ij}x_j
\end{equation}
for some scalar coefficients $\phi_{ij}$. By stacking all nodes' output in vector $\bby = [y_1, \ldots, y_N]^\transp$, we can write \eqref{eq.locFilt} as $\bby = \bbH(\bbS)\bbx$, where the matrix $\bbH(\bbS) : \reals^N \to \reals^N$ denotes the graph filter. The graph filter can be expressed as a function of the shift operator $\bbS$ in  different ways. Two widely used approaches\footnote{Recent works consider also more general approaches such as the node-variant \cite{SegarraMarques17} and the edge-variant graph filter\cite{CoutinoIsufiLeus2018}. To keep the exposition simple, we will discuss quantization of the two baseline approaches and leave the extension to the other methods for future research.} are the FIR graph filter \cite{sandryhaila2013discrete,SegarraMarques17} and the ARMA graph filter \cite{IsufiLoukasSimonettoLeus17,LiuIsufiLeus2019}.


\vskip.1cm
\textbf{FIR.} An FIR graph filter is a polynomial of order $K$ in the shift operator $\bbS$ with output:
\vspace{-0.18cm}
\begin{equation}\label{eq.FIRpol}
\bby =\bbH(\bbS)\bbx = \sum_{k = 0}^K\phi_k \bbS^k\bbx 
\vspace{-0.1cm}
\end{equation}
and scalar coefficients $\phi_0, \ldots, \phi_K$. The filtering behavior of $\bbH(\bbS)$ can be viewed by means of the GFT:
\vspace{-0.1cm}
\begin{equation}\label{eq.FIRfresp}
h(\lambda) = \sum_{k = 0}^{K}\phi_k\lambda^k\quad\for\quad\lambda \in [\lambda_{\min}, \lambda_{\max}]
\vspace{-0.1cm}
\end{equation}
which is a polynomial in the generic graph frequency $\lambda$. This spectral representation allows to define a filtering operator by specifying the analytic function $h(\lambda) : \lambda \to \reals$; hence, by approximating the latter with the polynomial in \eqref{eq.FIRfresp}, we can implement it distributively over the nodes through the recursion \eqref{eq.FIRpol} \cite{shuman2011distributed}. The distributed implementation is feasible because the shifted signal $\bbx^{(1)} = \bbS\bbx$ can be obtained through local exchanges between neighboring nodes in one communication round [cf. \eqref{eq.locFilt}]. The $k$th shifted signal can be obtained recursively as  $\bbx^{(k)} = \bbS\bbx^{(k-1)}$, where  nodes communicate to their neighbors the shifted signal  $\bbx^{(k-1)}$ obtained in the $(k-1)$th communication round. 
The output $\bby$ of the FIR graph filter is obtained after $K$ iterations of exchanges between neighbors, implying that in total, each node $i$ exchanges $K \text{deg}(i)$ messages with its neighbors. 
 This yields a communication complexity of order $\ccalO(MK)$.

\vskip.1cm
\textbf{ARMA.}
The ARMA graph filter extends \eqref{eq.FIRfresp} to a rational spectral response \cite{IsufiLoukasSimonettoLeus17}:
\vspace{-0.1cm}
\begin{equation}\label{eq.ARMAFresp}
\begin{split}
 h(\lambda) &= \frac{\sum_{q = 0}^Q b_q \lambda^q}{1 + \sum_{p = 1}^Pa_p\lambda^p} =\sum_{k = 1}^K\bigg(\frac{\varphi_k}{1 -\lambda\psi_k} \bigg) + \sum_{l = 1}^{L}\phi_l\lambda^l\\
           &\quad\qquad\qquad\qquad\qquad\qquad\quad\for\quad\lambda \in [\lambda_{\min}, \lambda_{\max}]
\end{split}
\vspace{-0.1cm}
\end{equation}
which allows for more flexibility when designing the filter coefficients $a_1, \ldots, a_P$ and $b_0, \ldots, b_Q$ (or the respective roots $\varphi_1, \ldots, \varphi_K$, poles $\psi_1, \ldots, \psi_K$, and direct term $\phi_1, \ldots, \phi_L$ coefficients) \cite{LiuIsufiLeus2019}. Without loss of generality, we consider $L = 0$ and refer to the filter in right-side  of \eqref{eq.ARMAFresp} 
   as an ARMA$_K$ graph filter \cite{IsufiLoukasSimonettoLeus17}. 



We can implement the ARMA$_K$ graph filter through the iterative recursion:
\vspace{-0.1cm}
\begin{align}\label{eq.ARMAK_iter}
\begin{split}
\bbw_t^{(k)} &= \psi_k\bbS\bbw_{t-1}^{(k)} + \varphi_k\bbx\\
\bby_t &= \sum_{k = 1}^{K}\bbw_t^{(k)}
\end{split}\quad \for~t \ge 1
\vspace{-0.2cm}
\end{align}
where $\bby_t$  is the ARMA$_K$ output at iteration $t$ and $\bbw_t^{(k)}$ is the output of the $k$th branch at iteration $t$ with arbitrary initialization $\bbw^{(0)}$. 
Recursion \eqref{eq.ARMAK_iter} builds the overall output $\bby_t$ at iteration $t$ as the sum of all $K$ parallel branches outputs $\bbw_t^{(k)}$ and converges ($t\to\infty$) to a steady-state only if the roots satisfy $|\psi_k| \le \rho $ for all $k = 1, \ldots, K$, where $\rho$ is the spectral radius of $\bbS$ \cite{IsufiLoukasSimonettoLeus17}.



The output of each branch $\bbw_t^{(k)}$ can be implemented distributively in a similar way as the FIR filters. The difference is that neighboring nodes exchange now the former output $\bbw_{t-1}^{(k)}$. Node $i$ combines the shifted outputs $w_{jt}^{(k)}$ from all neighbors $j \in \ccalN_i$ with its input signal $x_i$ with coefficients (as given in \eqref{eq.ARMAK_iter}) to obtain the output $w_{it}^{(k)}$. Finally, node $i$ combines locally all branches' outputs $w_{ti}^{(1)}, \ldots, w_{ti}^{(K)}$ to obtain the overall ARMA$_K$ output $y_{it}$ at iteration $t$. This procedure accounts for $K$ communication rounds between neighbors for each iteration $t$; hence, the overall communication cost of the ARMA$_K$  filter for $t = T$ iterations is of order $\ccalO(MKT)$.

Equations \eqref{eq.FIRpol} and \eqref{eq.ARMAK_iter} represent two fundamental algorithms to implement distributed GSP operations over WSNs. Our goal is to analyze the effects of dithered quantization to the filter outputs and account for it in the filter design phase. We shall analyze first quantization effects for static topologies in Sections~\ref{sec:FIRQeff} and \ref{sect-ARMA-quantization-effects} and later for random time-varying topologies in Section~\ref{sec.tvGraphs}. Before proceeding with this analysis for the FIR graph filters, let us briefly introduce the conceptual terminology of dithered quantization.

\vspace{-0.2cm}
\subsection{Dithered quantization}
Quantizing consists of encoding the data prior to its transmission with a certain number of bits, reducing the amount of information to be transmitted \cite{widrowkollar2008}.


Uniform quantizers map each input signal value   to the nearest value of a finite set of quantization levels, 
 where the quantization stepsize $\Delta$ between two adjacent levels is constant \cite{GrayNeuhoff1998}. We denote the quantized version of signal $\bbx$ as $\tilde{\bbx}=Q(\bbx)$ and it is given by:
\vspace{-0.15cm}
 \begin{equation}
  \tilde{\bbx} = \bbx +\bbn_\text{q}
  \vspace{-0.15cm}
\end{equation}
where $\bbn_\text{q}$ is the quantization noise. Although the quantization noise is deterministic, for a sufficiently small quantization stepsize $\Delta$, it can be well modeled as a uniformly random variable with zero-mean and variance $\Delta^2/12$, that is independent from the input \cite{SripadSnyder1977,LipshitzWannamaker1992}.

 To control the quantization noise and ensure the uniform random variable assumption
 and independence  from  the  input, we consider dithering quantization \cite{Schuchman1964,LipshitzWannamaker1992,SripadSnyder1977}. Dithering consists of adding a random additive signal  $\bbn_{\text{d}}$, called dither, to the input signal $\bbx$ prior to quantization. Dithering is widely used in distributed signal processing \cite{AysalCoates2008,KarMoura2010,ZhuChen2016, SchellekensShersonHeusdens2017}, which consists of iterative algorithms akin to distributed graph filtering.
In subtractive dithered quantization, the  dither signal is generated by a pseudo-random
generator at the transmitter node and it is subtracted at the receiving node after transmission. The receiver node
uses the  same  pseudo-random generator, which needs to be  agreed  prior to starting the communication.
Let us denote $\bbx_{\text{d}}=\bbx+ \bbn_{\text{d}}$ the dithered signal of $\bbx$. 
By applying quantization to the dithered signal $\bbx_{\text{d}}$, the transmitted signal  becomes: 
\vspace{-0.15cm}
\begin{equation}
\begin{aligned}
 \tilde{\bbx}_{\text{d}} & =Q(\bbx_{\text{d}})
  =Q(\bbx+\bbn_{\text{d}} )
   =\bbx+\bbn_{\text{d}} +\bbn_\text{q}=\tilde{\bbx}+\bbn_{\text{d}}
 \end{aligned}
 \vspace{-0.15cm}
\end{equation}
where signal $\tilde{\bbx}$ can be recovered by the receiver node by subtracting the dither $\bbn_{\text{d}}$ from the received signal $\tilde{\bbx}_{\text{d}}$.
 
The dither signal $\bbn_{\text{d}}$ follows an i.i.d. uniform distribution with first and second order moments:
\vspace{-0.15cm}
\begin{equation}
 \mathbb{E}[\bbn_{\text{d}}]  = \bb0 ~\text{and}~  \boldsymbol \Sigma_{{\text{d}}} =  \sigma^2_{{\text{d}}} \bbI=\frac{ \Delta^2}{12} \bbI.
 \label{statis-prop-dither}
 \end{equation}
The quantization noise $\bbn_\text{q}$ also follows a uniform distribution with statistical properties:
\vspace{-0.15cm}
\begin{equation}
 \mathbb{E}[\bbn_\text{q}]  = \bb0 ~\text{and}~ \boldsymbol \Sigma_{\text{q}} =  \sigma^2_{\text{q}} \bbI=\frac{ \Delta^2}{12} \bbI
 \label{statis-prop-quantiz}
 \end{equation}
and with realisations independent of the input. 
  
Two possible cases can be adopted when performing quantization with substractive dithering, namely, a constant quantization stepsize for all iterations or a dynamically decreasing quantization stepsize over the iterations, which offers a benefit as compared to a fixed quantization stepsize. Decreasing the quantization stepsize implies transmitting more bits over the iterations but this increase of communication overhead improves the control over the quantization noise. In the sequel, we will analyze both cases. 


\vspace{-0.2cm}
\section{FIR quantization analysis}\label{sec:FIRQeff}


This section analyzes the quantization effects in FIR graph filters.
We first discuss the fixed quantization stepsize and then the  dynamically decreasing stepsize. Next, we formulate a filter design problem that is robust to quantization noise.

\vspace{-0.2cm}
\subsection{Fixed  quantization stepsize}\label{sect-Fixed-Quantiz-FIR}

Consider the $k$th shifted signal $\bbx^{(k)}=\bbS^k \bbx$ exchanged with the neighbors. The quantized form of the latter is $\tilde{\bbx}^{(k)}=Q(\bbx^{(k)})=\bbx^{(k)}+ \bbn_\text{q}^{(k)}$. 
At the filter initialization, we have $\bbx^{(0)}=\bbx$, which quantized form is
$\tilde{\bbx}^{(0)}$ = $\bbx^{(0)}+ \bbn_\text{q}^{(0)}$.
This quantized signal is exchanged with neighbors leading to the quantized shifted signal
$ {\bbx}^{(1)}$ = $\bold S  \tilde{\bbx}^{(0)} =\bold S  (\bbx^{(0)} + \bbn_\text{q}^{(0)})$.
 Signal $\bbx^{(1)}$ is further quantized into $\tilde{\bold x}^{(1)}$ and subsequently transmitted to the neighboring nodes.
The process is repeated  $K$ times.
Based on the derivation in  Appendix \ref{Appendix-FIR-output}, the FIR filter output [cf. \eqref{eq.FIRpol}] with quantization becomes:
\vspace{-0.15cm}
\begin{equation}
 \begin{aligned}
  \bby^\text{q}     & = \sum_{k=0}^K \phi_k \bbS^k \bbx +  \sum_{k=1}^K \phi_k  \sum_{\kappa=0}^{k-1} \bbS^{k-\kappa } \bbn_\text{q}^{(\kappa)} \\
 \end{aligned}
 \label{output-FIR-quantiaz}
 \vspace{-0.15cm}
\end{equation}
 \noindent where  the second term on the right-hand side of (\ref{output-FIR-quantiaz}) accounts for the accumulated quantization error on the output:
 \begin{equation}
\begin{aligned}
\boldsymbol \epsilon  &= \bby^\text{q} - \bby
= \sum_{k=1}^K \phi_k  \sum_{\kappa=0}^{k-1} \bbS^{k-\kappa } \bbn_\text{q}^{(\kappa)}.
\end{aligned}
\label{quantiz-error}
\vspace{-0.15cm}
\end{equation}
We analyze next this quantization error in the spectral domain to ease the filter design.  
The following proposition provides a closed-form expression of the quantization noise mean squared error (MSE).

\begin{proposition} Consider the FIR graph filter of order $K$ in \eqref{eq.FIRpol} with coefficients $\phi_0, \ldots, \phi_K$ and quantization error $\bbepsilon$ in \eqref{quantiz-error} under fixed quantization stepsize. Consider also the graph Fourier transform $\hbepsilon = \bbU^{-1}\bbepsilon$ of the error  with respect to the shift operator $\bbS = \bbU\bbLambda\bbU^{-1}$. The average quantization  MSE per node $\hat{  \zeta}_\text{q} =  \mathbb{E}\left[\frac{1}{N} \text{tr}(\hat{\boldsymbol \epsilon} \hat{\boldsymbol \epsilon}^H )\right]$ is:
\vspace{-0.15cm}
 \begin{equation}
\begin{aligned}
\hat{  \zeta}_\text{q} & = \frac{\sigma^2_{\textnormal{q}}}{N}  \sum_{k=1}^K  \phi_{k}  \sum_{\kappa=0}^{k-1}  \| \boldsymbol \Lambda^{k-\kappa } \|_F^2.    \\
 \end{aligned}
 \label{MSE-error-freq-final}
 \vspace{-0.15cm}
\end{equation}
where $\|\cdot\|_F$ denotes the Frobenius norm and $\sigma_{\textnormal{q}}$ is the uniform quantizer standard deviation.
 \label{Proof-MSE-epsilon}
 \vspace{-0.15cm}
\end{proposition}
 \textit{Proof}: See Appendix \ref{Appendix-Proof-MSE-epsilon}.\qed

Proposition~\ref{Proof-MSE-epsilon} characterizes the impact of the graph frequencies $\bbLambda$ on the
quantization in FIR graph filters. A shift operator with large eigenvalues  amplifies the quantization MSE. This is because the high frequency terms  contribute more to the quantization noise. In other words, shift operators with small spectral radius bounds are preferred (e.g., normalized Laplacian or adjacency matrix). 
The filter coefficients $\phi_1, \ldots, \phi_K$
play also a role in the quantization error. As such, we can leverage expression \eqref{MSE-error-freq-final} to control the quantization MSE in the design phase, as suggested by \cite{ChamonRibeiro2017}. While expression (\ref{MSE-error-freq-final}) is useful if the eigendecomposition of the shift operator is computationally feasible, we can easily bound it by using the maximum eigenvalue. The
latter can be estimated with a lighter computational cost via power methods \cite{GolubvanderVorst2000}. 

\begin{corollaryP} Given the hypothesis of Proposition~\ref{Proof-MSE-epsilon} and also a shift operator $\bbS$ with maximum eigenvalue $\lambda_{\text{max}} \neq 1$,
the quantization MSE on the filter output $\hat{  \zeta}_\text{q}$ is bounded as:
\vspace{-0.15cm}
\begin{equation}
  \frac{\sigma^2_{\text{q}}}{N} \sum_{k=1}^K  \phi_{k}\;  \eta_k  \leq   \hat{  \zeta}_\text{q}  \leq  \sigma^2_{\text{q}} \!  \sum_{k=1}^K    \phi_{k} \;  \eta_k
  \label{bound-MSE-error-freq}
  \vspace{-0.15cm}
 \end{equation}
 where $\eta_k= ({  1-  \lambda_{\text{max}}^{2}  })^{-1}({  \lambda_{\text{max}}^2 - ( \lambda_{\max}^2)^{k+1}   })$.
  \label{bound-MSE-error-freq-corollary}
  \vspace{-0.15cm}
\end{corollaryP}
 \textit{Proof}: See Appendix \ref{Appendix-bound-MSE-error-freq-corollary}. \qed

The bounds in (\ref{bound-MSE-error-freq}) suggest that by working with a fixed quantization stepsize, the MSE has always a \emph{Cramer-Rao lower-bound} \cite{Kay1993} equivalence which cannot be overcome even by tuning the FIR coefficients in the design phase. In other words, even with a robust design strategy as the one in \cite{ChamonRibeiro2017}, we have an unavoidable error due to quantization that will affect the filter frequency response. To tackle this issue, we propose next an approach based on  dynamically decreasing the quantization stepsize,  which improves the control on the MSE. The caveat of this approach is that more bits are transmitted in the higher filter rounds ($k \to K$).

\vspace{-0.15cm}
\subsection{Dynamically decreasing  quantization stepsize}

Let us  consider a quantization stepsize $\Delta_{k}$ that decreases at each iteration $k$. That is, less bits are transmitted for earlier values of $k$ and more for $k \to K$. The main result is given by the following proposition.

\begin{proposition}\label{MSE-epsilon-DynCell} 
 Consider the FIR graph filter with shift operator $\bbS$ such that $\lambda_{\text{max}}>1$. 
 Consider also  the input signal is quantized with a uniform quantizer with decreasing quantization stepsize
$\Delta_{k}=(\lambda_{\text{max}})^{-k} \Delta_{0}$. Then,
the quantization MSE $\hat{  \zeta}_\text{q}$  of the FIR graph filter is upper bounded by\footnote{If the maximum eigenvalue $\lambda_{\text{max}}$ is exactly 1, we can add a small perturbation to it to make our assumption hold.}:
\vspace{-0.15cm}
\begin{equation} \label{DynCell-FIR}
 \hat{  \zeta}_\text{q}   \leq  \frac{\Delta_{0}^2}{12(1-( \lambda_{\text{max}}^{-2}))}  {\bold 1}^\top \boldsymbol \phi_{1}       
\vspace{-0.15cm}
 \end{equation}
where $\boldsymbol \phi_{1}=[\phi_{1}, \phi_{2}, \cdots, \phi_{K}]^\top$ is the vector that contains the FIR coefficients, except the term for $k = 0$. 
\vspace{-0.15cm}
 \end{proposition}
\textit{Proof}: See Appendix \ref{Appendix-MSE-epsilon-DynCell}. \qed

As opposed  to Proposition~\ref{Proof-MSE-epsilon}, expression \eqref{DynCell-FIR} shows that we have a clear control on the quantization MSE through $\bbphi_1$. Indeed, during the filter design phase, if we impose for the filter coefficients the condition that ${\bold 1}^\top \boldsymbol \phi_{1}\approx 0$, we can reduce significantly the quantization MSE.


\vspace{-0.2cm}
 \subsection{Filter design}

 Given a desired frequency response  $h^*(\lambda)$, we propose to design an FIR graph filter by solving the following convex optimization problem:
 \vspace{-0.15cm}
 \begin{equation}
\begin{array}{ll}
\underset{  \phi_{0}, \ldots, \phi_{K} } {\text{minimize}} &  \displaystyle \int_{\lambda} \left|  \sum_{k = 0}^{K}\phi_k \lambda^k - h^*(\lambda) \right|^2 \; d\lambda\\
  \text{subject to}    
                 &  \displaystyle \frac{1}{12} \sum_{k=1}^K \! \phi_{k}  \sum_{\kappa=0}^{k-1} \!   \Delta^2_{\kappa}   ( \lambda_{\text{max}}^2)^{k-\kappa} \leq          \epsilon\\    
                 &{\bold 1}^\top \boldsymbol \phi_{1}  \leq \gamma\\
                   &\Delta_k \le \delta
\label{optimiz-prob-FIR}  
\end{array}
\vspace{-0.15cm}
\end{equation}
For a finite small constant $\epsilon$, the first constraint controls the upper bound of the quantization MSE in case of both fixed and decreasing quantization stepsize [cf. \eqref{MSE-error-freq-DynCell}].
For an infinite value of $\gamma$, \eqref{optimiz-prob-FIR} leads to a similar optimization problem as \cite{ChamonRibeiro2017} for the case of fixed quantization stepsize,
while for the case of decreasing quantization stepsize, a finite small $\gamma$ can be used. 
In the last constraint,  $\delta$ controls the maximum quantization stepsize, 
implying hence the control of   the maximum number of bits  used at each iteration,  which we denote as $\chi$.


By subquantizing the initial data of $b_0$ bits with an average of $b$ bits at each iteration (i.e., $b<b_0$), the communication cost of FIR graph filter in term of bits exchanged reduces to $\ccalO(MKb)$.


\vspace{-0.2cm}
\section{ARMA quantization analysis}\label{sect-ARMA-quantization-effects}

This section analyzes the quantization effects on distributed  ARMA graph filters. Since ARMA filters reach the designed frequency response at steady-state, the signal quantization will have also  an effect on the filter convergence. We show in this section that the overall quantized MSE converges to zero if a dynamically decreasing quantization stepsize is considered, while this is not the case for the fixed stepsize-size quantizer.

\vspace{-0.2cm}
\subsection{Fixed  quantization stepsize} \label{sect-Fixed-Quantiz-ARMA}
Consider the parallel ARMA$_K$ graph filter in \eqref{eq.ARMAK_iter} and let us indicate by ${\bbw}_t^{\text{q}(k)} = Q(\bbw_t^{(k)}) = \bbw_t^{(k)} + \bbn_{t}^{\text{q}(k)}$ the quantized signal of branch $k$ at iteration $t$, i.e., $\bbw_t^{(k)}$. Here, $\bbn_{t}^{\text{q}(k)}$ denotes the respective quantization noise. 
Let also $\bbw_t=[{\bbw_t^{(1)}}^\top,{\bbw_t^{(2)}}^{\top},\cdots, {\bbw_t^{(K)}}^{\top} ]^{\top}$ be the $N K \times  1$ stacked vector containing all branches outputs and $\bbn_{t}^\text{q}=[{\bbn_t^{\text{q} (1)}}^\top,{\bbn_t^{\text{q} (2)}}^\top,\cdots, {\bbn_t^{\text{q}(K)}}^\top ]^\top$ the $N K \times  1$ stacked vector of quantization noise. Then, we can write the ARMA output $\bby_t$ due to quantization with the following compact notation: 
\vspace{-0.01cm}
\begin{align}\label{eq.ARMAK_iter_compact-Qunatiz-evolution}
\begin{split}
\bbw_t^\text{q} &= (\boldsymbol \Psi  \otimes \bbS) (\bbw_{t-1}^\text{q}+\bbn_{t-1}^\text{q} ) + \boldsymbol \varphi \otimes \bbx\\
\bby_t^\text{q} &= (\boldsymbol 1^\top  \otimes \bbI_N)  \bbw_t^\text{q}
\end{split}\quad \for~t \ge 1
\vspace{-0.15cm}
\end{align}
where $\otimes$ indicates the Kronecker product, 
$\boldsymbol \Psi=\text{diag}(\psi_1,\psi_2, \cdots, \psi_K)$ is the  $K \times K$ diagonal matrix
containing the former-output coefficients in the main diagonal and $\boldsymbol \varphi=[\varphi_1,\varphi_2,\cdots, \varphi_k]^\top$ is the
 $K \times 1$ coeffcient vector associated to the input.
By unfolding $\bbw_t^\text{q}$ in \eqref{eq.ARMAK_iter_compact-Qunatiz-evolution} to all its terms, we have:
\vspace{-0.15cm}
\begin{align}\label{eq.ARMAK_iter_compact-Qunatiz}
\bbw_t^\text{q} &=\!(\boldsymbol \Psi  \otimes \bbS)^t \bbw_{0} + \!\! \sum_{\tau = 0}^{t-1} (\boldsymbol \Psi  \otimes \bbS)^\tau (\boldsymbol \varphi \otimes \bbx) \!+ \!\sum_{\tau = 0}^{t-1} (\boldsymbol \Psi  \otimes \bbS)^{t-\tau} \bbn_{\tau}^\text{q}
\vspace{-0.15cm}
\end{align}
where the first two terms on the right-hand side account for the ARMA output up to iteration $t$, while the third term $ \boldsymbol \epsilon_t^{\textnormal{q}} = \sum_{\tau = 0}^{t-1} (\boldsymbol \Psi  \otimes \bbS)^{t-\tau} \bbn_{\tau}^\text{q}$ accounts for the accumulated quantization noise.
%
%

To analyze the MSE for the ARMA filter, let us first denote  by $\bbw^* = \lim_{t \to \infty}\bbw_t^\text{q}$ and by $\bby^* = \lim_{t \to \infty} \bby_t^\text{q}$ the steady-state values of $\bbw_t^\text{q}$ and $\bby_t^\text{q}$ in \eqref{eq.ARMAK_iter_compact-Qunatiz-evolution}, respectively. Let us also define the error:
 \vspace{-0.15cm}
 \begin{equation}
  \boldsymbol \epsilon_t^*=(\boldsymbol \Psi  \otimes \bbS)^t \bbw_{0}  + \sum_{\tau = 0}^{t-1} (\boldsymbol \Psi  \otimes \bbS)^\tau (\boldsymbol \varphi \otimes \bbx)   - \bbw^*
 \vspace{-0.15cm}
 \end{equation}
 which indicates how close the output of all branches $\bbw_t$ (without quantization) at iteration $t$ are w.r.t. the steady-state value $\bbw^*$. We  consider also the error $\boldsymbol \epsilon_{yt} =\bby_t^\text{q} - \bby^*$ between the quantized ARMA output $\bby_t^\text{q}$ in \eqref{eq.ARMAK_iter_compact-Qunatiz-evolution} and the steady-state output $\bby^*$, which can be written as follows: 
\vspace{-0.1cm}
 \begin{equation}
\begin{aligned}
\boldsymbol \epsilon_{yt} 
= (\boldsymbol 1^\top  \otimes \bbI_N)  \boldsymbol \epsilon_t^* + (\boldsymbol 1^\top  \otimes \bbI_N)  \boldsymbol \epsilon_t^\text{q} 
=  \boldsymbol \epsilon_{yt}^*+  \boldsymbol \epsilon_{yt}^\text{q}
\end{aligned}
\label{quantiz-error-ARMA}
\end{equation}
where $\boldsymbol \epsilon_{yt}^*=(\boldsymbol 1^\top  \otimes \bbI_N) \boldsymbol \epsilon_t^*$ indicates how close the \emph{unquantized} ARMA filter output $\bby_t$ at iteration $t$ is w.r.t.
its steady-state $\bby^*$ and $\boldsymbol \epsilon_{yt}^\text{q}=(\boldsymbol 1^\top  \otimes \bbI_N)  \boldsymbol \epsilon_t^\text{q}$ accounts for the propagation of the quantization noise over the iterations. 
 %
%
Then by simple algebra, the average MSE deviation per node of the error $\boldsymbol \epsilon_{yt}$ in \eqref{quantiz-error-ARMA} can be similarly split as:
 \vspace{-0.1cm}
 \begin{subequations}
\begin{align} \label{MSE-ARMA}
 {  \zeta}_{yt}&=\frac{1}{N} \mathbb{E}[ \text{tr}({\boldsymbol \epsilon}_{yt} {\boldsymbol \epsilon}_{yt}^\text{H} )]= {\zeta}_{yt}^*+{\zeta}_{yt}^\text{q}
\vspace{-0.15cm}
 \end{align}
 %
%
with:
\vspace{-0.15cm}
 \begin{align}
{  \zeta}_{y t}^*&=\frac{1}{N} \mathbb{E}[ \text{tr}((\boldsymbol 1^\top  \otimes \bbI_N) {\boldsymbol \epsilon}_{t}^* {{\boldsymbol \epsilon}_{t}^*}^\text{H} (\boldsymbol 1^\top  \otimes \bbI_N)^\text{H})]\\
{ \zeta}_{y t}^\text{q}&=\frac{1}{N} \mathbb{E}[ \text{tr}((\boldsymbol 1^\top  \otimes \bbI_N) {\boldsymbol \epsilon}_{t}^\text{q} {{\boldsymbol \epsilon}_{t}^\text{q}}^\text{H} (\boldsymbol 1^\top  \otimes \bbI_N)^\text{H})]
\label{MSE-ARMA-QQ}
\vspace{-0.15cm}
\end{align} 
\end{subequations}
where we have used  the linearity of the expectation w.r.t the trace 
 and  the independence of $\bbx$, $\bbw_0$ and $\bbn_{\tau}^\text{q}$.
${\zeta}_{yt}^*$ is the MSE for the case of unquantized filter from the steady-state output and ${\zeta}_{yt}^\text{q}$ is the quantization noise MSE at time $t$. 
 The following proposition provides an upper bound on the quantization MSE. 
 
 \vspace{-0.15cm}
 \begin{proposition}
Consider the ARMA$_K$ graph filter of order $K$ in \eqref{eq.ARMAK_iter_compact-Qunatiz-evolution} with coefficients $\boldsymbol \Psi$ and  $\boldsymbol \varphi$, and quantization error  $\boldsymbol \epsilon_{yt}^\text{q}$.
Let $\psi_\text{max}=\text{max}(|\psi_1|,|\psi_2|, \cdots, |\psi_K|)$ be the ARMA$_K$ coefficient with largest magnitude and let all ARMA$_K$ branches be stable i.e., $\psi_\text{max} \lambda_{\text{max}}  < 1$ for all $k=1 \cdots K$.
Consider also that the signal is quantized with a uniform quantizer with a fixed quantization stepsize   $\Delta$.
The quantization MSE ${\zeta}_{yt}^\text{q}$ of the filter at iteration $t$ 
is upper bounded by:
\vspace{-0.15cm}
\begin{equation}
{\zeta}_{yt}^\text{q}  \leq K  \sigma^2_{\text{q}}   \frac{  (\psi_\text{max}  \lambda_{\text{max}})^2 - \left((\psi_\text{max} \lambda_{\text{max}})^2\right)^{t+1}   }{  1-  (\psi_\text{max} \lambda_{\text{max}})^{2}  }.   
 \label{upperbound-MSE-error-ARMA}
 \vspace{-0.15cm}
 \end{equation}
Further, the steady-state ($t \rightarrow \infty$) quantization MSE is:
 \vspace{-0.15cm}
 \begin{equation}\label{eq.quantizMSE}
{\zeta}_{yt \rightarrow \infty}^\text{\text{q}}  \leq K  \sigma^2_{\text{q}}   \frac{  (\psi_\text{max}  \lambda_{\text{max}})^2   }{  1-  (\psi_\text{max} \lambda_{\text{max}})^{2}  }   
 \end{equation}\label{bound-MSE-error-ARMA}
 \label{Proposition-MSE-ARMA}
 \vspace{-0.15cm}
\end{proposition}
\textit{Proof}: See Appendix \ref{Appendix-Proposition-MSE-ARMA}. \qed

Proposition \ref{bound-MSE-error-ARMA} shows that the quantization MSE ${ \zeta}_{yt}^\text{q}$ of ARMA graph filters is upper bounded by a term that depends on the shift operator maximum eigenvalue. At steady-state $t \rightarrow \infty$, the overall ARMA MSE in \eqref{MSE-ARMA} is governed by the quantization MSE ${\zeta}_{yt}^\text{q}$ since the deviation ${\zeta}_{yt}$ from the steady-state  vanishes ${\zeta}_{yt \rightarrow \infty}^* \rightarrow 0$  for convergent stable filters. Therefore, we conclude that a fixed quantization stepsize heavily affects the ARMA filter behavior, which even at the steady-state, although not divergent, might lead to a completely different filtering behavior.

The filtering behavior of the ARMA recursion will not be considerably affected by the quantization noise in the early regime (i.e., small value of $t$) as long as:
\vspace{-0.1cm}
\begin{equation}
 {\zeta}_{yt}^*  \gg {\zeta}_{yt}^\text{q}. 
\end{equation}
However, for larger $t$, this inequality will be violated and the overall ARMA MSE will by dominated by  the quantization MSE ${\zeta}_{yt}^\text{q}$. While we might control \eqref{bound-MSE-error-ARMA} in the design phase of FIR graph filters, we should consider the challenges encountered when designing convergent distributed ARMA filters \cite{IsufiLoukasSimonettoLeus17}, i.e., the difficulty to guarantee an  accuracy-quantization robustness tradeoff. Rephrasing a non-convex design problem akin to \eqref{optimiz-prob-FIR} is possible, but because of non-convexity  that may lead to suboptimal design solutions,   in this work, we tackle this challenge by considering a decreasing quantization stepsize with $t$.

\vspace{-0.18cm}
\subsection{Dynamically decreasing  quantization stepsize}\label{Sect-decreas-quantiz-cell-ARMA}
Consider now a dynamic quantization stepsize size $\Delta_{t}$ that decreases with $t$ in a form that the quantization MSE ${\zeta}_{yt}^\textnormal{q}$ decreases with $t$ at least with the rate of the unquantized ARMA error $ {\zeta}_{yt}^*$ in \eqref{MSE-ARMA}. 
The following proposition shows this can be achieved. 
\vspace{-0.15cm}
\begin{theorem}
Consider the ARMA$_K$ graph filter of order $K$ in \eqref{eq.ARMAK_iter_compact-Qunatiz-evolution} with coefficients $\boldsymbol \Psi$ and  $\boldsymbol \varphi$, and quantization error  $\boldsymbol \epsilon_{yt}^\text{q}$.
Let $\psi_\text{max}=\text{max}(|\psi_1|,|\psi_2|, \cdots, |\psi_K|)$ be the ARMA$_K$ coefficient with largest magnitude and let all ARMA$_K$ branches be stable i.e., $\psi_\text{max} \lambda_{\text{max}}  < 1$ for all $k=1 \cdots K$. 
Consider also that the signal is quantized with a uniform quantizer with a   decreasing stepsize over the iterations $t$ as  $ \Delta_{t}=(\psi_\text{max} \lambda_{\text{max}})^{t} \Delta_{0}$. The quantization MSE $ {\zeta}_{yt}^\text{q}$ of the  filter output at iteration $t$ 
is upper bounded by:
\begin{equation}
 {\zeta}_{yt}^\text{q}   \leq   \frac{K  \Delta_{0}}{12} t (\psi_\text{max}  \lambda_{\text{max}})^{2t}     
\label{MSE-epsilon-DynCell-ARMA}
 \vspace{-0.15cm}
 \end{equation}
which at the steady-state  converges to zero (${\zeta}_{yt \rightarrow \infty }^\text{q}  \rightarrow 0$) at a rate of $t (\psi_\text{max}  \lambda_{\text{max}})^{2t}$. 
\label{Prop-MSE-epsilon-DynCell-ARMA}
 \vspace{-0.15cm}
 \end{theorem}

\textit{Proof}: See Appendix \ref{Appendix-Prop-MSE-epsilon-DynCell-ARMA}. \qed

Theorem \ref{Prop-MSE-epsilon-DynCell-ARMA} shows that by adopting a decreasing quantization stepsize, the quantization MSE for the ARMA filters vanishes at the steady-state. This behavior is similar to the convergence error of the unquantized ARMA ${\zeta}_{yt}^*$ and suggests that at steady-state, we can reach the designed filter response. However, the quantization MSE converges with a rate $t (\psi_\text{max}  \lambda_{\text{max}})^{2t}$ instead of $(\psi_\text{max}  \lambda_{\text{max}})^{2t}$. Faster convergence rates can be achieved by decreasing the quantization stepsize at a faster rate over time but this requires transmitting more bits  for larger values of $t$. 

Despite vanishing the quantization MSE at steady-state, the dynamic quantization stepsize comes together with a price. In particular, for large values of $t$, this implies that the quantization stepsize becomes infinitesimal; hence, the number of bits transmitted per round becomes that of the conventional ARMA graph filter [cf. \eqref{eq.ARMAK_iter}] after some iteration number $t \ge t^*$. Nevertheless, this strategy reduces the communication efforts in the first iterations, i.e., we can start with a coarser $\Delta_0$. For $b_t$ being the number of bits transmitted at iteration $t$, the communication cost of the ARMA$_K$ graph filter per iteration is of order $\ccalO(MKb_t)$. If $b$ is the average number of bits transmitted over $T$ iterations, the ARMA$_K$ communication complexity amounts to $\ccalO(MKTb)$. The benefits of following this approach is that the ARMA design is readily available from the unquantized setting \cite{IsufiLoukasSimonettoLeus17}.

A related problem that can be of  interest is to find the best sequence of quantization stepsizes $\Delta_0, \Delta_1, \cdots, \Delta_t$  by taking into account  the constraints of a given  total bit budget $\mathcal{B}$  available and  a  maximum  number of iterations $t_{\text{max}}$ and   where  $ \Delta_{t}=(\psi_\text{max} \lambda_{\text{max}})^{t} \Delta_{0}$.
 Note that the quantization stepsize $\Delta_t$   is defined as the ratio of the   quantization range $r_t$ at iteration $t$ over the number of quantization intervals given by
    $\Delta_t= {r_t}/{2^{b_t}}     ={ 2\;  \| \bold S^t \; \bold x \|_{\infty}} /{2^{{b_t}}}$,
 where   $\| \bold S^t \; \bold x \|_{\infty}$ is the maximum size of the messages  exchanged  between nodes at iteration $t$ and that
can be upper bounded by $ \| \bold S^t  \; \bold x  \|_{\infty} \leq \| \bold S^t   \|_2  \; \| \bold x \|_{2} \leq  {\rho}^t \; \| \bold x \|_{2} \leq  \| \bold x \|_{2}$
if the shift operator $\bold S$ has a spectral radius bound i.e., $\rho \leq 1$.
Thus, the best sequence of quantization stepsizes can be obtained for $\rho \leq 1$, $\psi_\text{max}  \lambda_{\text{max}} \neq 0$ and $  \psi_\text{max}  \lambda_{\text{max}} <1$ by solving the problem
$\sum_{t=0}^{t_{\text{max}}} \log_2 \left(\frac{2 \; \| \bold x \|_{2}}{(\psi_\text{max} \lambda_{\text{max}})^{t} \Delta_0}\right)=\mathcal{B}$, 
which implies  $\Delta_0=       2^{\left(1- \frac{\mathcal{B}}{1+t_{\text{max}}}   \right)}\; \| \bold x \|_2 \; (\psi_\text{max} \lambda_{\text{max}})^{-\frac{t_{\text{max}}}{2}}$.

\vspace{-0.001cm}
\section{Quantization analysis over time-varying graphs}\label{sec.tvGraphs}

We now extend the quantization analysis to cases where the graph connectivity changes randomly over the filtering iterations. This scenario is expected to occur in applications  of graph filtering over  WSNs. For our analysis, we consider directly the more general dynamically decreasing quantization stepsize and the random edge sampling model from \cite{IsufiLoukasSimonetto2017}. 
\vspace{-0.05cm}
\begin{definition}[Random edge sampling model \cite{IsufiLoukasSimonetto2017}]\label{def.RES}
 Consider an underlying graph $\mathcal{G}=(\mathcal{V}, \mathcal{E})$. A random  edge  sampling  (RES)  graph
 realization $\mathcal{G}_t=(\mathcal{V}, \mathcal{E}_t)$ of $\mathcal{G}$ is composed of the same set of nodes
 $\mathcal{V}$  and a random set of  links $\mathcal{E}_t \subseteq \mathcal{E}$  
 that are activated (i.e., $(i,j)\in  \mathcal{E}_t$)  with a probability $p_{ij}$ ($0 < p_{ij}\leq 1$).
 The links  are activated independently over  the graph and time and 
 are mutually independent from the graph signal.
 \label{def-res-model}
 \vspace{-0.15cm}
\end{definition}

We consider the RES graph realization to model the link losses that occur at each filter iteration. As such, the RES model states that  the realization $\mathcal{G}_t=(\mathcal{V}, \mathcal{E}_t)$ at iteration $t$ is drawn 
from the underlying connectivity graph $\mathcal{G}=(\mathcal{V}, \mathcal{E})$, where
 the links $\mathcal{E}_t \subseteq \mathcal{E}$ are generated via an i.i.d. Bernoulli process with  probability $p_{ij}$. 
Let then $\bold P \in \mathbb R^{N \times N} $ denote the matrix that collects the link activation probabilities $p_{ij}$. 
Let also $\bold S$, $\bold S_t$, and $\bar{\bold S}$ denote, respectively, the shift operator of the underlying graph $\mathcal{G}$, the graph realization $\mathcal{G}_t$ at iteration $t$, and the expected graph $\bar{\mathcal{G}}$. Since graph $\mathcal{G}$ has an upper bounded shift operator $\|\bbS\|_2 \le \rho$,  all its realizations $\mathcal{G}_t$ have also  an upper bounded shift operator $\| \bold S_t  \|_2  {\leq} \| \bold S  \|_2  {\leq} \rho$  \cite{GamaIsufiLeus2018,HOPPENMonsalve2019}.

Before, we proceed with the filter analysis, the following remark is in order. Under  the  RES model, 
$\text{if} \; \bold S= \bold A \;\text{then}$ the expected shift operator is $\bar{\bold S}=\mathbb{E}[\bold A_t]=\bold P \circ \bold A$.
If  $\bold S= \bold L$, then the expected shift operator\footnote{Note that if  $\bold P$ has equal rows so that  $p_{ij}=p_i \; \forall \; j \; \in \; \mathcal{V}$ or has equal entries i.e. $p_{ij}=p$, we have $\mathbb{E}[\bold L_t]=\bold P \circ \bold L$.}
 is $\bar{\bold S}=\mathbb{E}[\bold L_t]=  \bar{\bbD} - (\bold P \circ \bold A)$, where  $\bar{\bbD}=\mathbb{E}[\bold D_t]$ is a  diagonal matrix whose non zero  entries are given by $[\bar{\bbD}]_{ii}=\sum_{j=1}^{N} a_{ij} p_{ij}$.

\vspace{-0.25cm}
\subsection{FIR graph filters}\label{sec.tvGraphs-FIR}

When the FIR filter is run over RES graph realizations, the instantaneous shift operator $\bbS_t$ is present in the filtering expression \eqref{eq.FIRpol} and affects the output.
To characterize this output, let us define the transition matrix of the RES graph realisations $\ccalG_t, \ldots, \ccalG_{t^\prime}$, 
$\bold \Theta(t',t) = \prod_{\tau=t}^{t'}  \bold S_{\tau} \; \text{if} \; t'\geq  t$ 
and $\bold I  \; \text{if} \; t' <  t$.
The FIR filter output over a sequence of $K$ time-varying graphs is: 
\vspace{-0.18cm}
\begin{equation}
 \bold y_{t} =\sum_{k=0}^{K} \phi_k \;\bold \Theta(t-1,t-k) \; \bold x
\label{time-var-NVGF}
\vspace{-0.05cm}
 \end{equation}
where the filter output is computed by considering all graph realizations from the  iteration  $t {-} K$ to $t$.
From the independence of RES graph realizations, the expected FIR output is: 
\vspace{-0.18cm}
\begin{equation}\label{eq.FIRTV_mean}
\begin{aligned}
  \bar{\bold y}_{t} &= \mathbb{E} \big[ \bold y_t \big]=\displaystyle\sum_{k=0}^{K}  \phi_k   {\bar{ \bold S}}^{k}  \bold x 
\end{aligned} 
\vspace{-0.05cm}
\end{equation}

As shown in  Appendix \ref{Appendix-FIR-output-varying-graph}, the quantized FIR filter output over RES graph realizations can be written as $ {\bold y}_{t}^\text{q}  =\bby_t +\boldsymbol \epsilon_t$
where the quantization error $\boldsymbol \epsilon_t$ has the expression: 
\vspace{-0.1cm}
\begin{equation}
 \boldsymbol \epsilon_t=   \sum_{k=1}^{K}  \sum_{\kappa=0}^{k-1}   \phi_k  \bold \Theta(t{-}\kappa -1,t{-}k) \;  \bbn_\text{q}^{(\kappa)}.
 \end{equation}
The latter accounts for the percolation of the quantization noise $\bbn_\textnormal{q}^{(\kappa)}$ over different random graph realizations.
%
%
Since the quantization noise has a zero mean, the expected FIR output with quantization is $ \mathbb{E} \big[ {\bold y}_t^\text{q} \big] = \bar{\bold y}_{t}$ [cf. \eqref{eq.FIRTV_mean}]. That is, in expectation, the FIR graph filter behaves as the filter in \eqref{time-var-NVGF} operating on the expected graph with unquantized data.
%

To quantify the statistical impact of the quantization noise, we analyze the second order moment of the quantized output $\bby_t^\text{q}$ in the following proposition. 
\vspace{-0.01cm}
\begin{proposition}
Consider the FIR graph filter operating over the RES graph realizations $\ccalG_t$ [cf. Def. \ref{def.RES}] with shift operators $\bbS_t$ upper bounded as $\|\bbS_t\|_2 \le \rho$. Let also the filter input signal be quantized with a dynamic  quantization stepsize size $\Delta_t$ at iteration $t$. The MSE of the filter output due to quantization and graph randomness $\zeta^\textnormal{q}_t=  \mathbb{E}[\frac{1}{N} \text{tr}({\boldsymbol \epsilon_t} {\boldsymbol \epsilon_t}^H )]$ is upper bounded by:
\vspace{-0.15cm}
\begin{equation}
\begin{aligned}
  \zeta^\text{q}_t   &  \leq   \frac{1}{12} \sum_{\kappa=1}^{K} \Delta_{\kappa-1}^2  \bigg(\sum_{k=\kappa}^{K} \rho^{k-\kappa+1} | \phi_k | \bigg)^2.\\
  \end{aligned}
 \label{eq_prop18bound}
\end{equation}
\label{Prop-MSE-Vary-graph}
\vspace{-0.15cm}
\end{proposition}

\textit{Proof}: See Appendix \ref{Appendix-Prop-MSE-Vary-graph}. \qed
 
Note that Proposition \ref{Prop-MSE-Vary-graph} represents the worst-case bound for the graph randomness. This is similar to the unquantized graph filters over RES graphs \cite{IsufiLoukasSimonetto2017}  because the spectral radius $\rho$ accounts for all potential link losses (it is independent on the probabilities $p_{ij}$). On the other hand, this result serves  as a proxy for the MSE to design a graph filter that is robust to both link losses and quantization error. 

\smallskip
\noindent \textbf{Filter design.} 
Our goal is to design the filter coefficients $\phi_0, \ldots, \phi_K$ to reduce the quantization MSE in \eqref{eq_prop18bound} while keeping the quantized graph filter output $ \bby_{t}^\text{q} $ close in expectation to the unquantized output over the deterministic graph $\ccalG$; we denote the latter as $\bby^\diamond=\sum_{k = 0}^K \phi_k^\diamond \bbS^k\bbx$. 
   Then, let us consider the expected error due to quantization (bias):
\begin{equation}
 \bar{\bold e}=\mathbb{E} \big[ \bby_{t}^\text{q}   -  \bold y^\diamond \big]=\mathbb{E} \big[ \bby_{t}^\text{q}  \big] -  \bold y^\diamond. 
 \end{equation}
While we can design the coefficients to minimize this bias, they will not account for the deviation around it. Therefore, we consider the more involved problem of  finding the filter coefficients as a  trade-off between 
 the expected error of the filter output and the quantization MSE. For this, let us define the filtering matrix difference $ \bar {\bold E}$: 
 \vspace{-0.15cm}
  \begin{equation}\label{eq:matrixU}
\bar {\bold E}= \sum_{k=0}^{K} \left(  \phi_k  \; {\bar{ \bold S}}^{k}   -   \phi_k^\text{$\diamond$}   \; \bold S^{k}  \right) 
 \vspace{-0.15cm}
 \end{equation}
 that accounts for the response difference between the graph filtering over the expected graph $\bar{\mathcal{G}}$ and the graph filtering over the deterministic graph $\mathcal{G}$. Then, we find the filter coefficients by solving the convex problem: 
\vspace{-0.15cm}
 \begin{equation}
\begin{array}{ll}
\underset{ \phi_k}{\text{minimize}}   \left  \| \bar {\bold E} \right \|_F   +  \displaystyle \frac{\gamma}{12}  \sum_{\kappa=1}^{K} \Delta_{\kappa-1}^2  \bigg(\sum_{k=\kappa}^{K} \rho^{k-\kappa+1} | \phi_k | \bigg)^2
\label{optimiz-prob}
\end{array}
\vspace{-0.15cm}
\end{equation}
\noindent where  $\|  \bar {\bold E} \|_F$ is the Frobenius norm of \eqref{eq:matrixU}
 and $\gamma$ is a weighting factor trading-off the expected error and the quantization MSE.


 \vspace{-0.35cm}
\subsection{ARMA graph filters}
The parallel ARMA filter operating over random graphs has the branches outputs: 
\vspace{-0.1cm}
\begin{align}\label{eq.ARMAK_iter_compact-time-varying}
\begin{split}
\bbw_t &= (\boldsymbol \Psi  \otimes \bbS_{t-1}) \bbw_{t-1} + \boldsymbol \varphi \otimes \bbx\\
\end{split}
\vspace{-0.15cm}
\end{align}
%
%
which in the presence of quantization noise becomes:
\vspace{-0.1cm}
\begin{align}\label{eq.ARMAK_iter_quantiz-time-varying-compact}
\begin{split}
\bbw_t^\text{q} &= (\boldsymbol \Psi  \otimes \bbS_{t-1}) (\bbw_{t-1}^\text{q}+\bbn_{t-1}^\text{q} ) + \boldsymbol \varphi \otimes \bbx.\\
\end{split}
\vspace{-0.1cm}
\end{align}
By expanding \eqref{eq.ARMAK_iter_quantiz-time-varying-compact} to all the terms, we can write the overall ARMA filter output due to quantization as: 
\vspace{-0.15cm}
\begin{align}\label{eq.ARMAK_iter_compact-Quantiz-time-varying}
\begin{split}
\bbw_t^\text{q} \! &= \!  \big(\!  \prod_{\tau=0}^{t-1}\! \boldsymbol \Psi  \!\otimes \!\bbS_{\tau} \big) \bbw_{0} \!+\!\boldsymbol \varphi \!\otimes \!\bbx \!+ \!\! \sum_{\tau=1}^{t-1} \! \! \bigg(\! \prod_{\tau^\prime=t-\tau}^{t-1} \!  \boldsymbol \Psi  \! \otimes \!\bbS_{\tau^\prime} \bigg)\!  (\boldsymbol \varphi \!\otimes \!\bbx) 
\! + \!\boldsymbol \varepsilon_t^\text{q}\\
\bby_t^\text{q} &= (\boldsymbol 1^\top  \otimes \bbI_N)  \bbw_t^\text{q}
\end{split}
\vspace{-0.15cm}
\end{align}
where in order to ease notation, we have denoted by $\boldsymbol \varepsilon_t^\text{q} = \sum_{\tau=0}^{t-1} \big(\prod_{\tau^\prime=\tau}^{t-1}  \boldsymbol \Psi  \otimes \bbS_{\tau^\prime} \big) \bbn_\tau^\text{q}$ the percolation of the quantization noise $\bbn_\tau^\text{q}$ over the parallel ARMA branches up to time $t$.
%
Then, let us consider the filter output error $\boldsymbol \varepsilon_{yt}=\bby_t^\text{q} - \bby^*$ from the steady-state expected ARMA output $\bby^*$:
\vspace{-0.15cm}
\begin{equation}
\begin{aligned}
\boldsymbol \varepsilon_{yt}=   \boldsymbol \varepsilon_{yt}^\text{q} + \boldsymbol \varepsilon_{yt}^*
\vspace{-0.15cm}
\end{aligned}
\label{quantiz-error-ARMA-time-varying}
\end{equation}
 where $\boldsymbol \varepsilon_{yt}^\text{q}=(\boldsymbol 1^\top  \otimes \bbI_N)  \boldsymbol \varepsilon_t^\text{q}$ is the quantization error on the output; $\boldsymbol \varepsilon_{yt}^*=(\boldsymbol 1^\top  \otimes \bbI_N)  \boldsymbol \varepsilon_t^* $ is the unquantized ARMA graph filter error at iteration $t$ w.r.t. to its steady-state $\bby^*$. Then, let us denote by $\boldsymbol \varepsilon_t^*$  the unquantized ARMA error  w.r.t. to the steady-state $\bbw^*$, which is given by:
 \vspace{-0.15cm}
 \begin{equation}
 \boldsymbol \varepsilon_t^*=\!\bigg(\!\prod_{\tau=0}^{t-1}  \! \boldsymbol \Psi  \otimes \bbS_{\tau} \bigg) \bbw_{0} +\boldsymbol \varphi \otimes \bbx + \!\sum_{\tau=1}^{t-1} \!\!\bigg(\!\!\prod_{\tau^\prime=t-\tau}^{t-1} \!\!\! \boldsymbol \Psi  \otimes \bbS_{\tau^\prime}\!\! \bigg) (\boldsymbol \varphi \otimes \bbx) - \bbw^*.
\vspace{-0.15cm}
 \end{equation}
%

  Under the RES graph model and  given the zero-mean quantization noise, 
  it can be easily shown from  \eqref{eq.ARMAK_iter_compact-time-varying} and \eqref{eq.ARMAK_iter_quantiz-time-varying-compact} 
that $\mathbb{E}[{\bbw}_t^\text{q}]=\mathbb{E}[{\bbw}_t]$; i.e., in expectation both the quantized and unquantized ARMA filters  give the same output.  
However, the quantization  impacts on the second order moment of the filter output error $\boldsymbol \varepsilon_{yt}$
in \eqref{quantiz-error-ARMA-time-varying}.
We analyze  next  the    MSE  of the latter, which  by simple algebra, can be split as: 
\vspace{-0.15cm}
  \begin{subequations}
\begin{align} \label{MSE-ARMA-time-varying}
 {  \xi}_{yt}&= \frac{1}{N} \mathbb{E}[ \text{tr}({\boldsymbol \varepsilon}_{yt} {\boldsymbol \varepsilon}_{yt}^\text{H} )]= {\xi}_{yt}^*+{\xi}_{yt}^\text{q}.
 \vspace{-0.15cm}
 \end{align}
 where:
 \vspace{-0.15cm}
  \begin{align}
{  \xi}_{y t}^*&=\frac{1}{N} \mathbb{E}[ \text{tr}((\boldsymbol 1^\top  \otimes \bbI_N) {\boldsymbol \varepsilon}_{t}^* ({\boldsymbol \varepsilon}_{t}^*)^\text{H} (\boldsymbol 1^\top  \otimes \bbI_N)^\text{H})]\\ 
{ \xi}_{y t}^\text{q}&=\frac{1}{N} \mathbb{E}[ \text{tr}((\boldsymbol 1^\top  \otimes \bbI_N) {\boldsymbol \varepsilon}_{t}^\text{q} ({\boldsymbol \varepsilon}_{t}^\text{q})^\text{H} (\boldsymbol 1^\top  \otimes \bbI_N)^\text{H})]\label{MSE-ARMA-time-varying-QQ}
\vspace{-0.15cm}
\end{align}
\end{subequations}
and where we have used  the linearity of the expectation w.r.t the trace, the  cyclic property of the trace,  
 and  the independence of $\bbx$, $\bbw_0$ and $\bbn_{\tau}^\text{q}$.
 ${\xi}_{y t}^*$ is the MSE for the case of unquantized filter  w.r.t. to its steady-state output.
The next proposition provides an upper bound on  the quantization MSE  ${\xi}_{yt}^\text{q}$, when the quantization stepsize $\Delta_{k}$ decreases at each iteration $k$.

\begin{theorem}
Consider the  ARMA$_K$ graph filter operating over RES graph realizations $\ccalG_t$ [cf. Def. \ref{def.RES}]  with shift operators $\bbS_t$ upper bounded as $\|\bbS_t\|_2 \le \rho$.
Let $\psi_\text{max}=\text{max}(|\psi_1|,|\psi_2|, \cdots, |\psi_K|)$ be the ARMA$_K$ coefficient with largest magnitude and let all ARMA$_K$ branches be stable i.e., $\psi_\text{max} \;\rho  < 1$ for all $k=1 \cdots K$.
Let also the filter input signal be quantized with a uniform quantizer 
having a stepsize decreasing over the iterations $t$ as  $\Delta_{t}=(\psi_\text{max}  \;\rho )^{t} \Delta_{0}$. 
 The MSE of the  ARMA filter output at iteration $t$  due to quantization and graph randomness $ {\xi}_{y t}^\text{q}$ can be upper bounded by:
\vspace{-0.15cm}
 \begin{equation}
 {\xi}_{y t}^{\text{q}}   \leq   \frac{K^2\;  \Delta_{0}}{12} t \;( \psi_\text{max}  \; \rho )^{2t}     
\label{MSE-epsilon-DynCell-ARMA-TVGraph}
\vspace{-0.15cm}
\end{equation}
 making the quantization MSE converge to zero (${\xi}_{y t \rightarrow \infty }^\text{q}  \rightarrow 0$)  at a rate of $t (\psi_\text{max}   \;\rho)^{2t}$.
\label{Prop-MSE-epsilon-DynCell-ARMA-TVGraph}
\vspace{-0.15cm}
 \end{theorem}
 \textit{Proof}: See Appendix \ref{Appendix-Prop-MSE-epsilon-DynCell-ARMA-TVGraph}. \qed

 
  Theorem \ref{Prop-MSE-epsilon-DynCell-ARMA-TVGraph} highlights that  the quantization MSE converges to zero
 when using a decreasing quantization stepsize, despite the random topological changes and the presence  of quantization.
 This implies that there is no  need to consider the quantization MSE in the design phase.
 However, contrarily to  time-invariant graphs,  the overall MSE of ARMA filters, which is affected by both the quantization ${\xi}_{yt}^\text{q}$ and the random variation part ${\xi}_{yt}^\text{*}$, can not reach the desired filter response  at steady-state ($t \rightarrow \infty$), because
 even if the quantization MSE ${\xi}_{yt}^\text{q}$ can be made to converge  to zero, the unquantized MSE ${\xi}_{yt}^\text{*}$ does not converge  to zero due to graph topological changes. 

 Similarly to time-invariant graphs in Section~\ref{Sect-decreas-quantiz-cell-ARMA}, where we consider the constraints of a given   total bit budget $\mathcal{B}$  available and  a  maximum  number of iterations $t_{\text{max}}$,  the best sequence of quantization stepsizes  is
  given by  $\Delta_0=       2^{\left(1- \frac{\mathcal{B}}{1+t_{\text{max}}}   \right)}\; \| \bold x \|_2 \; (\psi_\text{max} \rho)^{-\frac{t_{\text{max}}}{2}}$ and   $\Delta_{t}=(\psi_\text{max}  \;\rho )^{t} \Delta_{0}$ for $\rho \leq 1$ and $ 0 < \psi_\text{max}  \;\rho  < 1 $.

  \begin{corollaryT}
Consider  same settings as Theorem \ref{Prop-MSE-epsilon-DynCell-ARMA-TVGraph} with 
 the  input signal  quantized with a uniform quantizer having a  fixed quantization stepsize  $\Delta$.
The MSE of the filter output due to quantization and graph randomness ${\xi}_{y t}^\text{q}$ can be upper bounded by: 
\vspace{-0.15cm}
\begin{equation}   \label{upperbound-MSE-error-ARMA-time-varying}
{\xi}_{y t}^\text{q}  \leq K^2  \sigma^2_{\text{q}}   \frac{  ( \psi_\text{max} \;\rho)^2 - [ (\psi_\text{max} \;\rho)^2]^{t+1}   }{  1-  (\psi_\text{max} \;\rho)^{2}  }   
 \vspace{-0.15cm}
 \end{equation}
 \noindent which in the steady-state ($t \rightarrow \infty$) becomes:
 \vspace{-0.15cm}
 \begin{equation}
{\xi}_{y t \rightarrow \infty}^\text{q}  \leq K^2  \sigma^2_{\text{q}}   \frac{  (\psi_\text{max}  \;\rho)^2   }{  1-  (\psi_\text{max} \;\rho)^{2}  }   
 \vspace{-0.18cm}
 \end{equation}
 \label{Proposition-MSE-ARMA-time-varying}
\end{corollaryT}

 \begin{proof}
 By considering a fixed quantization stepsize $\Delta$,  the  upper bound of the  MSE of ARMA filter due to quantization and graph randomness in \eqref{Proof-MSE-ARMA-quantiz-time-varying-final-dyn} becomes: 
 \vspace{-0.18cm}
 \begin{equation}
\begin{aligned}
{ \xi}_{y t}^\text{q} & \leq K^2 \sigma_\text{q}^2  \sum_{\tau=1}^{t}  \big(\psi_\text{max} \; \rho\big)^{2\tau}   
\end{aligned}
\label{Proof-MSE-ARMA-quantiz-time-varying-final}
\vspace{-0.15cm}
\end{equation}

\noindent  By considering
the upper bound in \eqref{Proof-MSE-ARMA-quantiz-time-varying-final} is   finite geometric series  with argument smaller than 1, 
${ \xi}_{y t}^\text{q}$  can be upper bounded by \eqref{upperbound-MSE-error-ARMA-time-varying}.
\vspace{-0.15cm}
 \end{proof}



%

\begin{figure}[t]
\begin{center}
\includegraphics[scale=0.38]{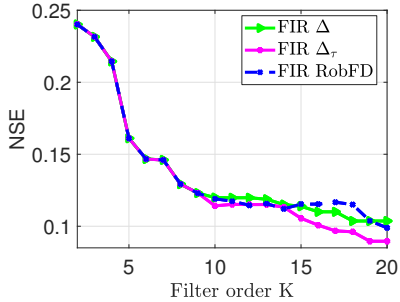}
\end{center}
\caption{NSE  of  FIR graph filters over time-invariant graphs, when approximating an ideal low-pass filter. 
 The  filter coefficients  are optimized by solving \eqref{optimiz-prob-FIR}, where the maximum number of bits at each iteration is $\chi=32$ bits. 
 The results are compared to the Robust  Filter  Design (RobFD) solution \cite{ChamonRibeiro2017}.
}
 \label{fig-lowPassFilterFIR}
 \vspace{-0.2cm}
\end{figure}

\vspace{-0.2cm}
\section{Numerical Experiments}\label{sect-results}
This section corroborates our theoretical findings with numerical  experiments on both synthetic and real data from 
the Molene and the Intel Berkely sensor network.

\vspace{-0.2cm}
\subsection{Synthetic data}
We consider   WSNs with $N=100$  sensor nodes, which are randomly and uniformly distributed over
 a square  area of side $150$ m.
 Each node can communicate with the  neighbors within the transmission range $R=50$ m. The latter forms a communication network that can be used to perform distributed graph filtering operations. In the sequel, we evaluate the quantized filters in the baseline ideal-low pass filter and signal denoising. To account for the graph randomness, we averaged the results over 1000 different realizations.

\smallskip
\noindent\textbf{Ideal low-pass filter.}
We considered the FIR graph filter to approximate an ideal low-pass filter with 
 frequency response $h(\lambda)=1$ if $  \lambda \leq \lambda_c$ and zero otherwise. The shift operator is the normalized Laplacian $\bold \bbL_\text{n}=\bbD^{-1/2} \bbL \bbD^{-1/2}$ and the cut off frequency $\lambda_c$ is half the spectrum. The input signal $\bold x$ has a white unitary spectrum w.r.t. the underlying graph $\mathcal{G}$.

 \begin{figure}[t]
\begin{center}
\includegraphics[scale=0.38]{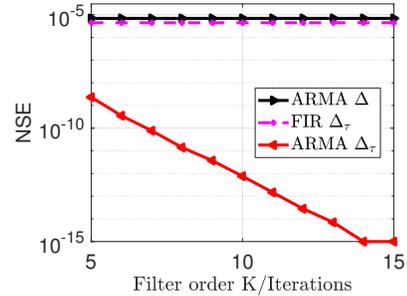}
\end{center}
\caption{ NSE  between the quantized output 
and the unquantized output of FIR and ARMA  filtering over time-invariant graphs for the Tikhonov denoising problem, where  $\bold S=\bold \bbL_\text{n}$ and  $w=0.3$.
 The FIR filter coefficients are optimized by solving \eqref{optimiz-prob-FIR},  with $\chi=25$ bits. 
The x-axis is the filter order for the FIR filter and the number of iterations for the ARMA filter.
 }
\label{fig-FIRARMAFixDynStepDenoisTIVG}
\vspace{-0.15cm}
\end{figure}

   \begin{figure*}[t]
\begin{center}
\subfloat[$p=0.95$, $\chi=15$ bits]{\includegraphics[scale=0.38]{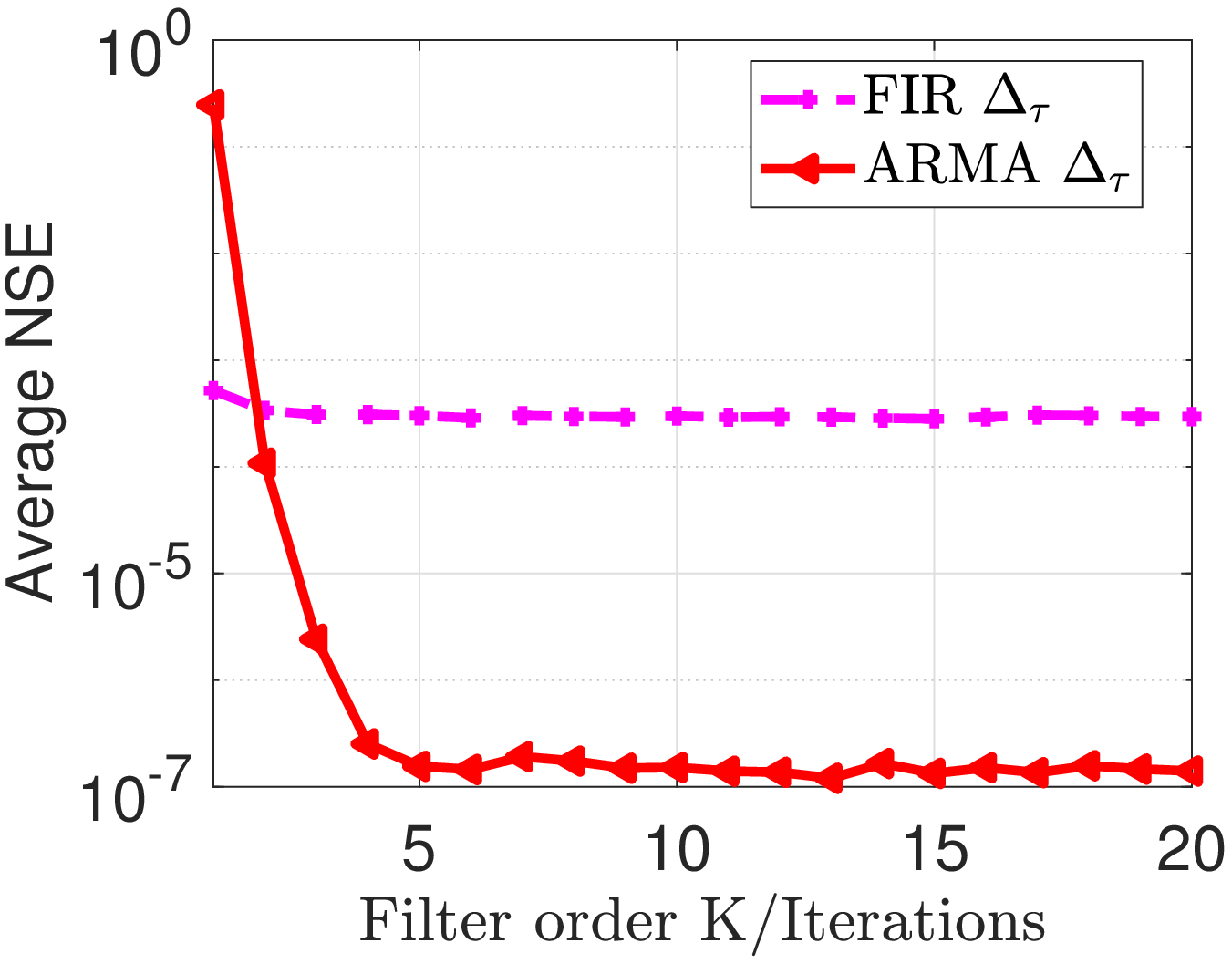}\label{fig-FIRARMADynStepDenoisTVaryG}}
\hspace{0.5cm}\subfloat[$\chi=15$ bits]{\includegraphics[scale=0.38]{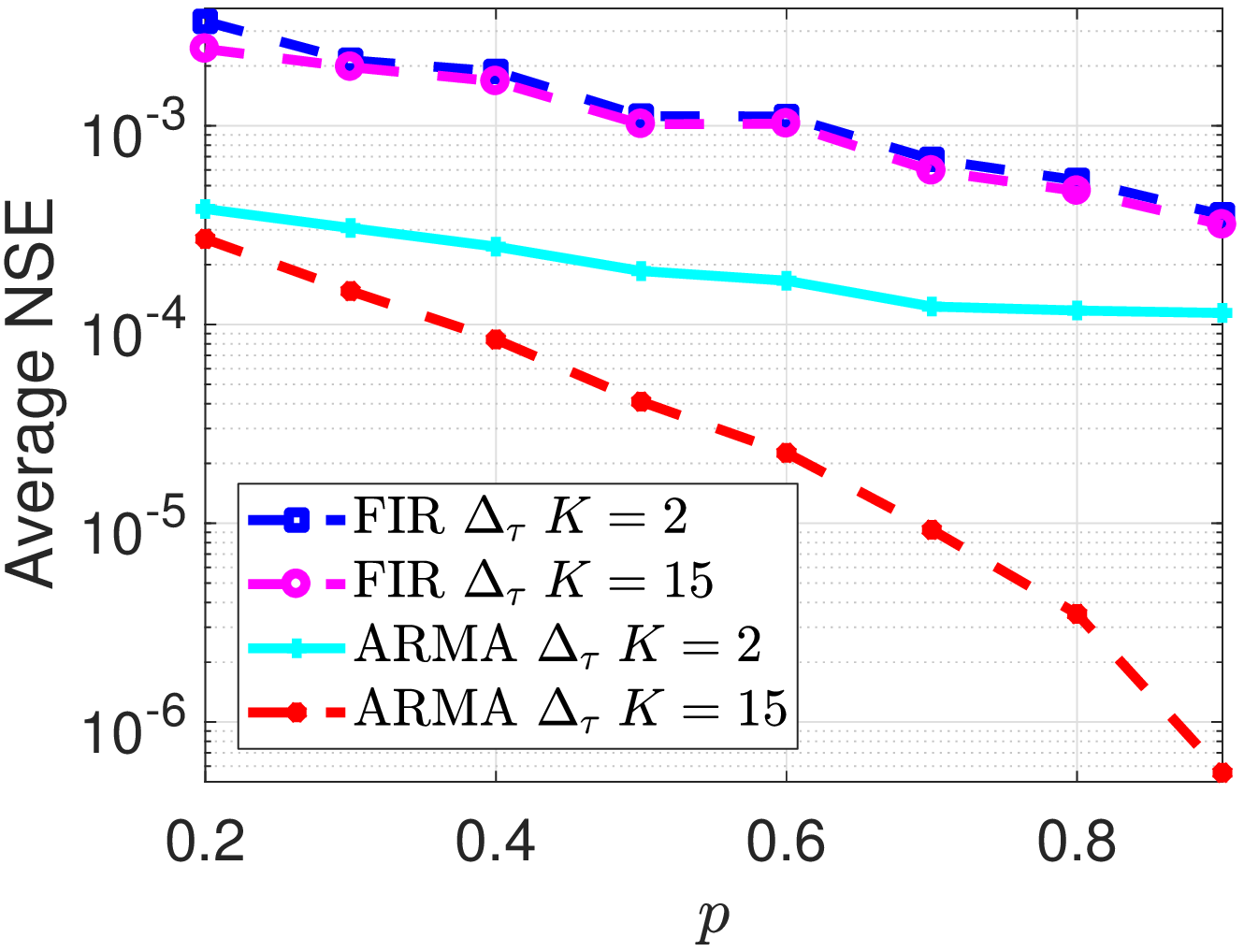}\label{fig-FIRARMADynStepDenoisTVaryGprob}}
\hspace{0.5cm}\subfloat[$K=10$]{\includegraphics[scale=0.38]{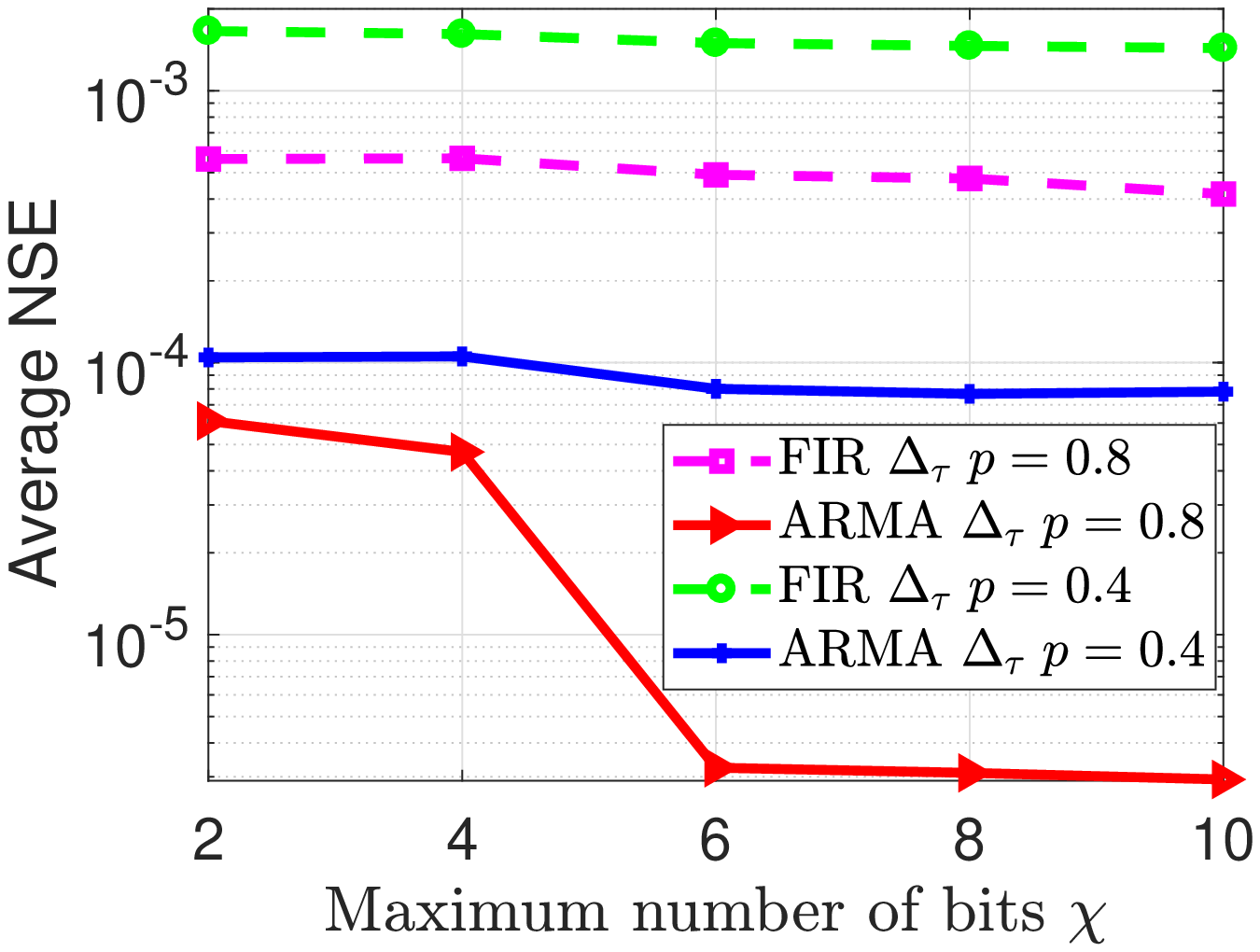}\label{fig-FIRARMADynStepDenoisTVaryGMaxBit}}
\end{center}
\caption{
(a)(b)(c) Average NSE  between the  quantized output over time-varying graph
and the unquantized output over deterministic graph of  FIR and ARMA  filters for Tikhonov denoising problem, where $\bold S=\lambda_{\text{max}}^{-1}\bbL$ and $w=0.25$.
 The FIR filter coefficients are optimized by solving \eqref{optimiz-prob}.
(a) Average NSE  vs. filter order $K$/iterations,  where $p=0.95$ and $\chi=15$ bits. The x-axis is the filter order for FIR filter and the number of iterations for  ARMA filter.
(b) Average NSE  vs.  the probability of link activation $p$, where $\chi=15$ bits.
(c) Average NSE  vs.  the maximum number of bits $\chi$  at each iteration, where $K=10$.
}
\label{fig-denoising-TVaryG}
\vspace{-0.15cm}
\end{figure*}
 
  \begin{figure*}[ht]
\begin{center}
\subfloat[]{\includegraphics[scale=0.38]{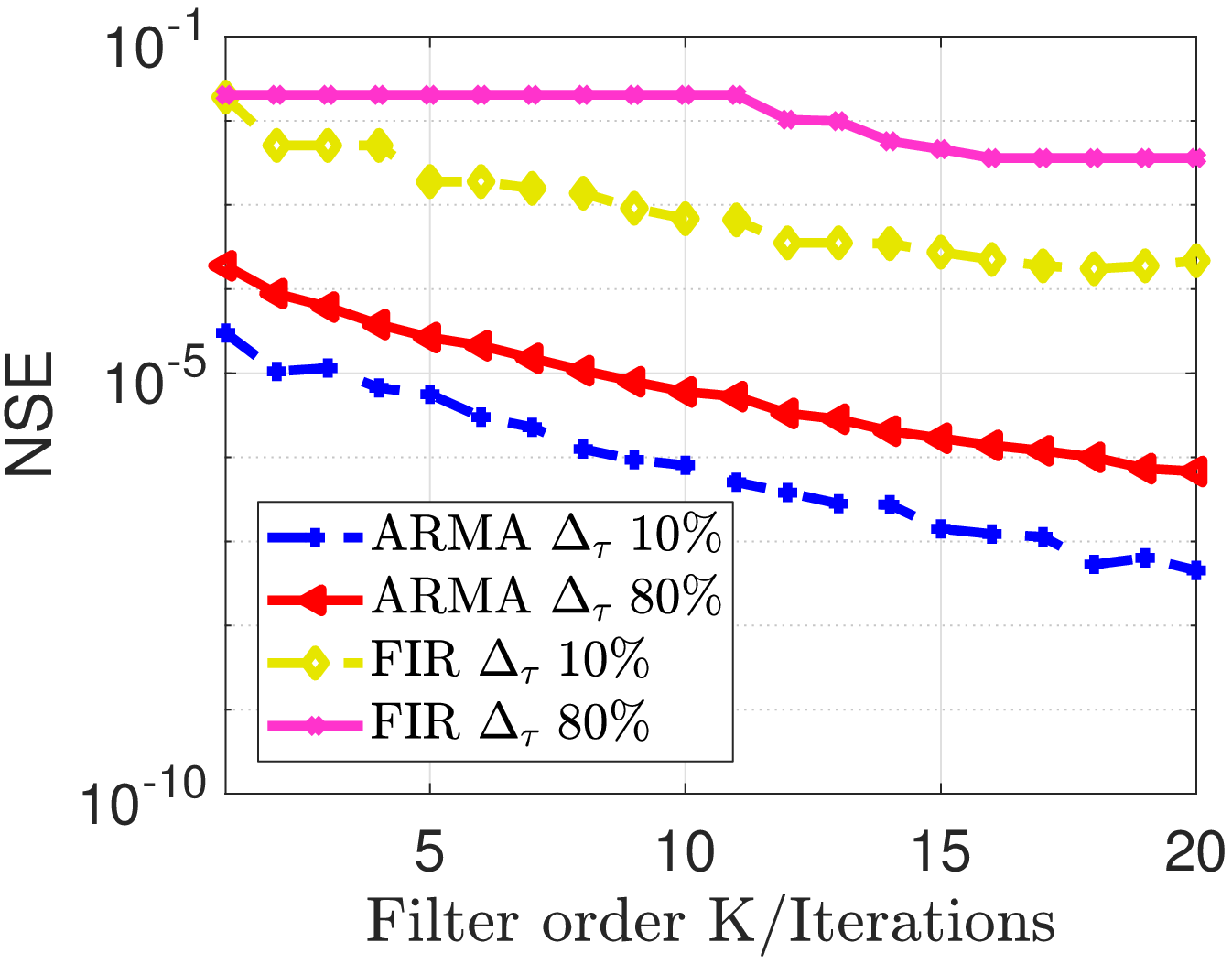}\label{fig-realDataMoleneBrestInterpolation}}
\hspace{0.15cm}\subfloat[]{\includegraphics[scale=0.38]{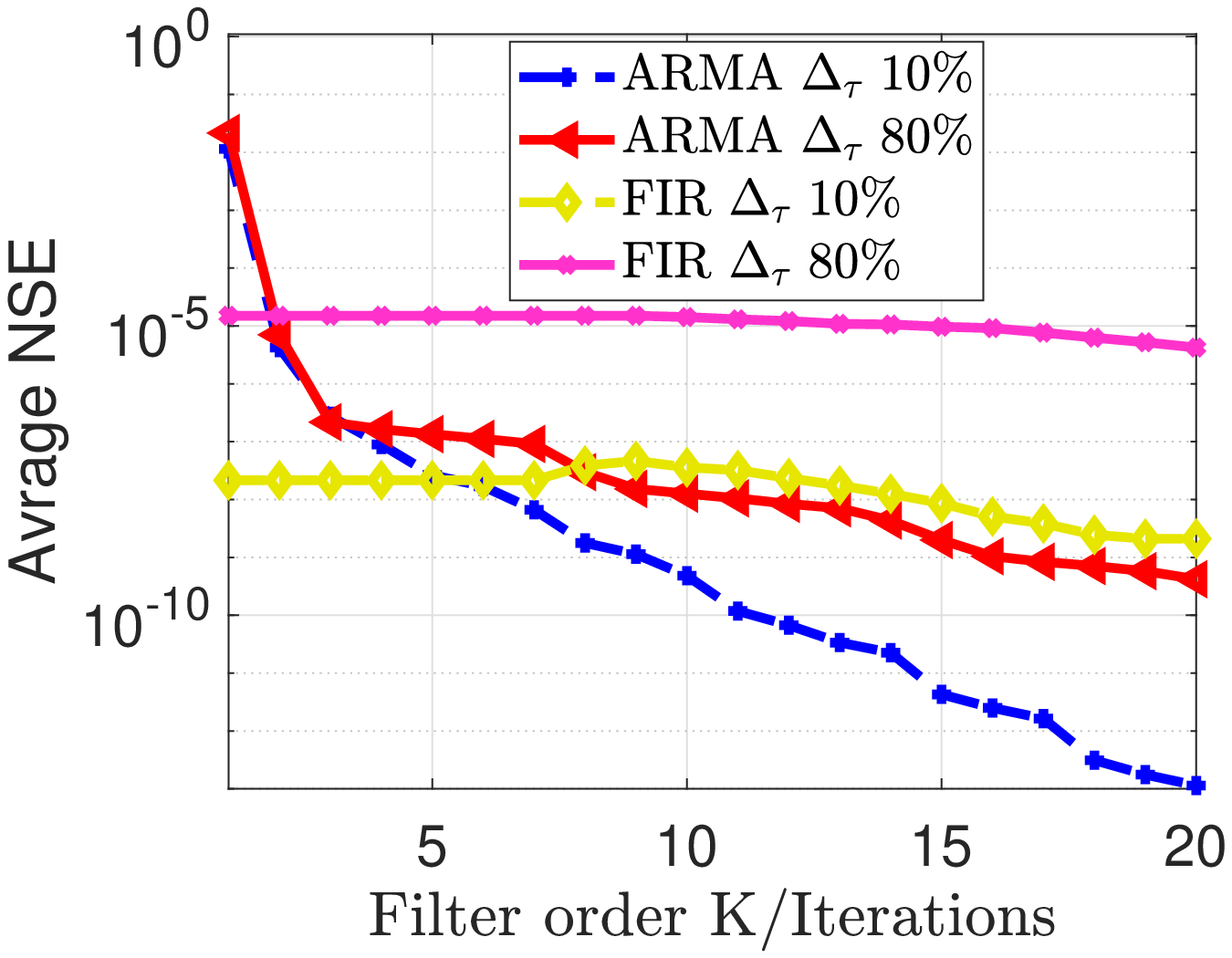}\label{fig-realDataIntelDataInterpolation}}
\hspace{0.1cm}\subfloat[]{\includegraphics[scale=0.38]{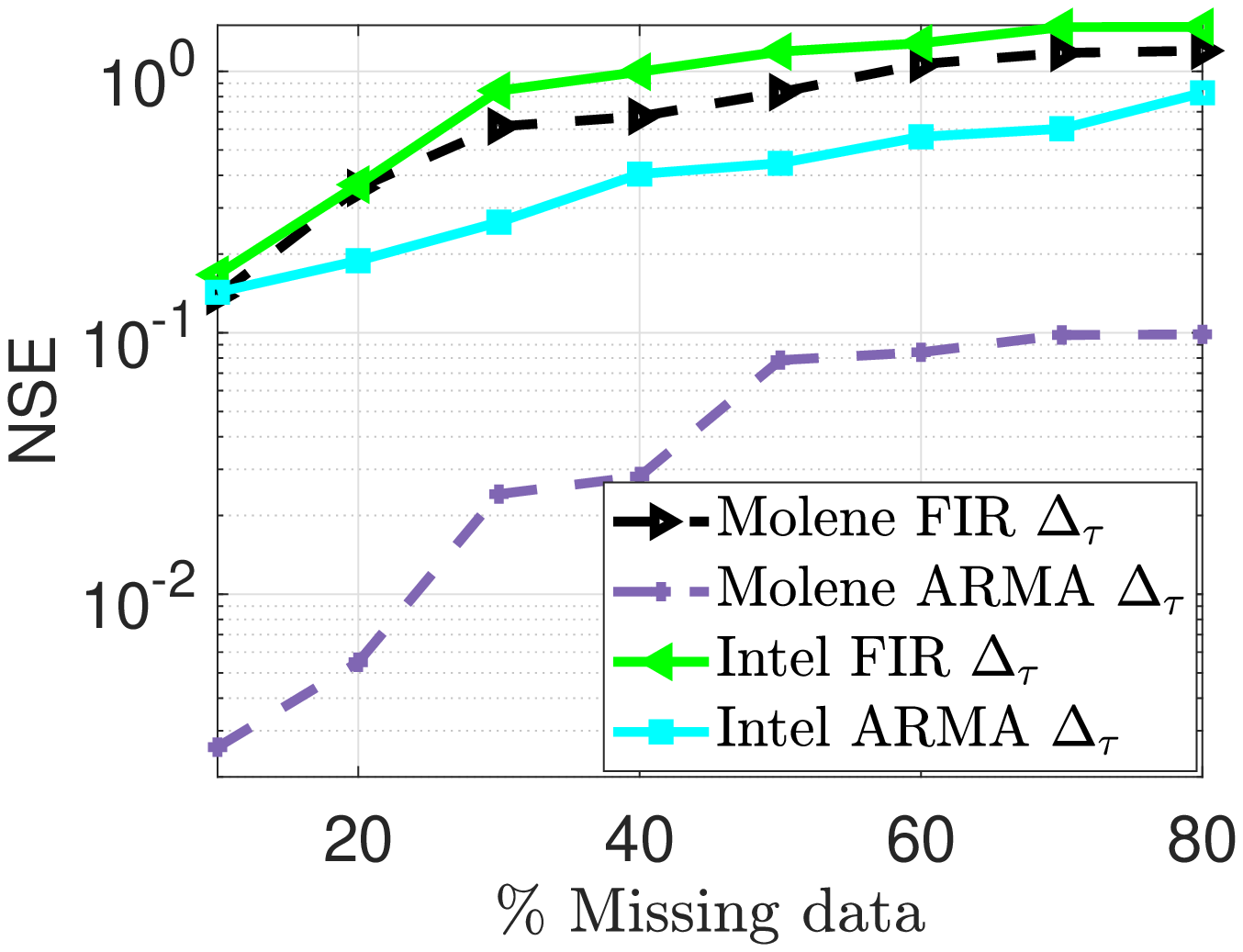}\label{fig-realDataReconstruction}}
\end{center}
\caption{
(a)  NSE  between the  quantized output 
and the unquantized output of  graph  filters, when interpolating  the missing temperature values in Molene weather data set
and  where $\bbS=\bbL_\text{n}$.  
The FIR filter coefficients are optimized by solving \eqref{optimiz-prob-FIR}.
(b) Average NSE  between the  quantized output 
and the unquantized output of  graph  filters, when interpolating the missing light values in Intel
  Lab  data set  and where $\bbS=\lambda_{\text{max}}^{-1}\bbL$.
The FIR filter coefficients are optimized by solving \eqref{optimiz-prob}.
(a)-(b) The x-axis is  the number of iterations for  ARMA filter. 
(c) NSE  between the  quantized output and the graph signal to be reconstructed for different data set vs. percentage of missing data after $20$ filter iterations.
(a)(b)(c) The parameters are $w=0.3$ and $\chi=15$ bits.
}
\label{fig-interpolation-realData}
\vspace{-0.15cm}
\end{figure*}

Fig. \ref{fig-lowPassFilterFIR} shows the  Normalized  Squared Error NSE $=\| \hby^{\text{q}} -  \hby \|_2^2/ \| \hby \|_2^2$
between the quantized output  $\hby^{\text{q}}$ and the unquantized desired signal  $\hby$
in the graph frequency domain, when the FIR filters run over time-invariant graphs. 
The filter coefficients of FIR graph filters with fixed  and decreasing quantization stepsizes are optimized by solving \eqref{optimiz-prob-FIR}.
The results show that both  our designed FIR graph filters and the Robust  Filter  Design (RobFD)  \cite{ChamonRibeiro2017} achieve similar 
performance  for low filter orders ($K<10$). However, the FIR graph filter with decreasing quantization stepsize 
performs better than the other two alternatives for higher filter orders.

\smallskip
\noindent\textbf{Tikhonov denoising.}
We now evaluate the performance of the proposed solutions in distributed denoising.
We assume a noisy graph signal $\bold x=\bold z+ \bold n$, where $\bold z$ is the  signal of interest and 
 $\bold n$ is a zero mean additive noise.
 To recover  signal $\bold z$, 
 we solve the  Tikhonov denoising problem:
 \vspace{-0.05cm}
 \begin{equation}\label{eq_Tikhonov_denoising}
 \bold z^{*}=  \underset{\bold z \in {\mathbb{R}}^N} {\text{argmin}} \| \bold x -\bold z    \|_2^2  + w \; \bold z^\top  \bold S \bold z
 \end{equation}
 for   $\bold S=\bbL$ or $\bbS=\bbL_\text{n}$, and where the regularizer $\bold z^\top  \bold S \bold z$ is  based on the prior assumption  the graph signal varies  smoothly with  respect to the underlying graph and $w$ is the weighting factor trading smoothness and noise removal \cite{Shuman2013}.
 The  closed-form solution of (\ref{eq_Tikhonov_denoising})  is an  ARMA$_{1}$  filter $\bold z^{*}=(\bbI+w\bbS)^{-1} \bbz$ with coefficients $\psi=-w$ and $\varphi=1$ \cite{IsufiLoukasSimonettoLeus17}. Hence, we can employ said filter to solve distributively the Tikhonov denoising problem. 

In Fig. \ref{fig-FIRARMAFixDynStepDenoisTIVG}, we compare the NSE$=\| \bby_t^{\text{q}} -   \bby_t \|_2^2/ \| \bby_t \|_2^2$ between the quantized  
and the unquantized outputs  of FIR  and  ARMA graph filters over time-invariant graphs. The noise in this instance is zero-mean Gaussian with variance $\sigma^2=0.2$.
We observe the ARMA graph filter with decreasing quantization stepsize significantly outperforms both the ARMA with fixed quantization stepsize and
the FIR graph filter with optimized filter coefficients and decreasing quantization stepsize. The latter corroborates our finding in Theorem~\ref{Prop-MSE-epsilon-DynCell-ARMA}: ARMA filters reach machine precision with a decreasing quantization stepsize.


 We now evaluate the filters over time-varying graphs,
 by analyzing the average NSE between the  quantized output over the time-varying graph $\bby^\text{q}_t$ and the unquantized output $\bby_t$ over the deterministic graph.
 As shown in Fig. \ref{fig-denoising-TVaryG}~\subref{fig-FIRARMADynStepDenoisTVaryG},
 the ARMA graph filter presents  significantly better performance than   
the FIR graph filter,  when the link activation probability is $p=0.95$
and the quantization stepsize is decreasing over the iterations.
We also observe the average NSE for ARMA filters
reduces considerably when the number of iterations increases.
Notice also the NSE floor of ARMA filter is the value when the signal is quantized with all the available  bits and where  $\Delta_k$ is very small. The latter corresponds to the machine precision accuracy, corroborating our results in Theorem~\ref{Prop-MSE-epsilon-DynCell-ARMA-TVGraph}.



In Figs.~\ref{fig-denoising-TVaryG}~\subref{fig-FIRARMADynStepDenoisTVaryGprob}-\subref{fig-FIRARMADynStepDenoisTVaryGMaxBit}, we analyze  the average NSE  
for different  probabilities of link activation and different maximum numbers  of bits used  for quantization. 
ARMA filters with decreasing quantization stepsize achieves always the highest filtering accuracy with a significant margin compared to other filters.
Fig.~\ref{fig-denoising-TVaryG}~\subref{fig-FIRARMADynStepDenoisTVaryGprob}  shows that, as expected, better link connectivities (higher $p$)
lead to lower errors as expected.
It is also worth noticing that the graph filtering accuracy is less affected by topological changes due to link losses 
 for lower filter orders $K$,  as compared to higher filter orders. This is because the exchanges between nodes through problematic links reduce.
This highlights the trade-off between the filter order and robustness to topological changes: a higher order should be preferred when the topology is highly stable.
Fig.~\ref{fig-denoising-TVaryG}~\subref{fig-FIRARMADynStepDenoisTVaryGMaxBit} shows  that  the average NSE  decreases when the maximum number  of bits used  for quantization at each iteration is higher. 
 This because increasing the quantization bits decreases the quantization stepsize at each iteration, which  reduces as well the quantization errors accumulated among iterations.
We can also observe that  increasing the quantization bits does not lead
necessarily to a noticeable improve of the  filtering accuracy, especially  for low  probability of link activation, as compared to higher  probability of link activation. 
We attribute this behavior to the large number of links that fall, therefore, the error due to link losses dominates that of quantization.

\vspace{-0.15cm}
 \subsection{Real data}
 We now illustrate the performance of the proposed solutions for the graph signal 
 interpolation task over time-invariant and time-varying topologies with two real data sets.
 
 \smallskip
\noindent\textbf{Molene weather data set.}
The Molene weather data contains hourly observations of temperature
 measurements of $N=32$ weather stations collected  in the region of
 Brest in France, for a total of $744$ hours. 
 We consider a geometric graph  constructed from the node coordinates using
the default nearest neighbor approach, as in \cite{IsufiLoukasPerraudin2019}.


Let  $\bbx'$ be the observed graph signal $\bbx$ with missing values.
We aim at reconstructing the overall graph signal $\bbx$ from the 
observations $\bbx'$ by exploiting the smoothness of $\bbx$ over the graph. This problem can be formulated
as \cite{narang2013signal,mao2014image}: 
\vspace{-0.15cm}
 \begin{equation}\label{eq_interpolation}
 {\bbx}^\star= \underset{ {\bold x \in {\mathbb{R}}^N}}    {\text{argmin}} \| \bbT (\bold x -\bold x')    \|_2^2  + w \; \bold x^\top  \bold S \bold x
 \vspace{-0.15cm}
 \end{equation}
 where $\bbT$ is a diagonal matrix with $\bbT_{ii} = 1$ if $x_i$ is known and
$\bbT_{ii} = 0$ otherwise and  $w$ is the weighting factor.
The optimal solution of the convex optimization problem \eqref{eq_interpolation} is 
${\bbx}^\star=(\bbT+w\bbS)^{-1} \bbx'=(\bbI - \tilde{\bbS})^{-1} \bbx'$, which is an
 ARMA$_1$ filter for the shift operator $\tilde{\bbS}=\bbT+w\bbS-\bbI$ \cite{IsufiLoukasSimonettoLeus17}.
 To generate missing  values in the Molene weather data set, we randomly wipe off signal values up to certain percentage.
 Then, we analyze the NSE  between the quantized output 
and the unquantized output of   graph filters,
for different percentages of  missing values.

Fig. \ref{fig-interpolation-realData}~\subref{fig-realDataMoleneBrestInterpolation} shows the NSE decreases considerably at each iteration, particularly for ARMA filters.
It is also  worth noticing  this decrease  enhances  when less data are missing.


\smallskip
\noindent\textbf{Intel Lab data set.}
The Intel Berkeley Research Lab data set contains light data of $N=54$   Mica2Dot sensor nodes  distributed in an indoor environment over an area of 1200 $\text{m}^2$ \cite{BodikHong2004}.
The  communication between the sensor nodes  is wireless and prone to   channel noise and interference, leading to  time-varying  graph topological changes due to link losses.
The  probability of link activation of the nodes is about 0.13 with a standard deviation of 0.18.
The underlying graph topology has  high connectivity with an  average node degree  of  47, implying    multiple
communication paths  exist between nodes, helping to make signal exchanged between nodes robust to link losses.

In Fig. \ref{fig-interpolation-realData}~\subref{fig-realDataIntelDataInterpolation}, 
we analyze the average quantized NSE as a function of the missing values for the FIR and ARMA graph filters.
Even though  the graph filtering accuracy is affected by the accumulated quantization errors over iterations and the graph topological changes, ARMA filters provide a significant
decrease in terms of NSE, when the number of iterations grows and the percentage of missing data is low.
Fig. \ref{fig-interpolation-realData}~\subref{fig-realDataReconstruction} depicts the NSE between 
the quantized graph signal output and the true signal
for the two  data sets  as a function of the missing values. 
The results show that for both data sets a quite good performance in terms of signal reconstruction is achieved,  especially
with ARMA graph filters. This stresses our finding in Theorem~\ref{Prop-MSE-epsilon-DynCell-ARMA-TVGraph}, which means that with a decreasing quantization stepsize there is no need to perform a robust ARMA graph filter design since the proposed strategy achieves the optimal steady-state solution.



\vspace{-0.25cm}
\section{Conclusion}\label{sect-conclusion}
In this work, we provided a broader analysis of the  quantization effects of both FIR and ARMA graph filters over time-invariant and time-varying graphs. 
We analyzed the impact of fixed and dynamic quantization stepsize on the filtering performance.
For  FIR filters, we first showed  that a dynamic quantization stepsize leads to a  more control on the quantization MSE  than fixed-stepsize quantization
and then we proposed a robust filter design that minimizes the quantization noise.
For ARMA graph filters, we showed that 
decreasing  the  quantization stepsize over iterations  reduces the quantization MSE to zero at steady-state.
We extended our quantization effects analysis of FIR and ARMA  graph filters  to  networks affected by random  topological changes due to link losses and  propose a novel filter design strategy that is robust to quantization and topological changes.
Extensive numerical experiments with synthetic and real data show the  different trade-offs between 
quantization bits, filter order, and robustness to topological randomness,  ultimately, highlighting the efficiency of the proposed solutions.


As our work puts a new practical paradigm for distributed aspects of graph filters, we identify as relevant future research direction the application of these filters for digital and distributable graph neural networks, network coding, and finite-time consensus.

\vspace{-0.15cm}
\section{Appendix}
%
\subsection{Quantized FIR graph filter output}\label{Appendix-FIR-output}

Considering  ${\bold x}^{(0)} =\bold x $ and the  quantized message at iteration $k$,
$\tilde{\bold x}^{(k)}= {\bold x}^{(k)}+ \bbn_\text{q}^{(k)}$, the output of the shifted signal with quantization is:
\vspace{-0.1cm}
\begin{small}
\begingroup
\allowdisplaybreaks
\begin{align}\label{quantiaz-vector-exch}
  {\bold x}^{(1)} &= \bold S \tilde{\bold x}^{(0)}  = \bold S  ({\bold x}^{(0)}{+} \bbn_\text{q}^{(0)})= \bold S  {\bold x}^{(0)}{+} \bold S \bbn_\text{q}^{(0)} \nonumber\\
  {\bold x}^{(2)} &= \bold S \tilde{\bold x}^{(1)} 
  =  \bold S^2 {\bold x}^{(0)}{+}   \bold S^2 \bbn_\text{q}^{(0)}{+}  \bold S \bbn_\text{q}^{(1)} \nonumber\\
  \vdots  & \nonumber  \\
 {\bold x}^{(k)} &= \;   \bold S^k   \; \bold x^{(0)} + \sum_{\kappa=0}^{k-1}  \;  \bold S^{k-\kappa}  \; \bbn_\text{q}^{(\kappa)}, \; \; k \geq 1.
  \end{align}
  \endgroup
\end{small} 
\vspace{-0.25cm}

From (\ref{quantiaz-vector-exch}), the FIR graph filter output with quantization is:
\vspace{-0.1cm}
\begin{small}
\begin{equation}
 \begin{aligned}
& \bby_\text{q}  =\phi_0 \bbx 
        +   \phi_1 (\bbS \bbx + \bbS \bbn_\text{q}^{(0)})
        + \phi_2 (\bbS^2 \bbx+ \bbS^2 \bbn_\text{q}^{(0)} + \bbS \bbn_\text{q}^{(1)} )
        + \cdots \\
       & + \phi_k (\bbS^{K} \bbx + \bbS^{K} \bbn_\text{q}^{(0)}  + \bbS^{K-1} \bbn_\text{q}^{(1)} + \cdots + \bbS^2 \bbn_\text{q}^{(K-2)}
        + \bbS \bbn_\text{q}^{(K-1)} )\\
       & = \sum_{k=0}^K \phi_k \bbS^k \bbx +  \sum_{k=1}^K \phi_k  \sum_{\kappa=0}^{k-1} \bbS^{k-\kappa } \bbn_\text{q}^{(\kappa)}. \\
 \end{aligned}
\end{equation}
\end{small}
\vspace{-0.5cm}

\vspace{-0.2cm}
\subsection{Proof of Proposition \ref{Proof-MSE-epsilon}}\label{Appendix-Proof-MSE-epsilon}
 By  applying the GFT on both sides of (\ref{quantiz-error}),  the quantization error has the spectral response:
\vspace{-0.1cm}
\begin{equation}
\begin{aligned}
 \hat{\boldsymbol \epsilon}  &= \sum_{k=1}^K \phi_k  \sum_{\kappa=0}^{k-1} \boldsymbol \Lambda^{k-\kappa }   \hat{\bbn}_\text{q}^{(\kappa)}  \\
\end{aligned}
\label{quantiz-error-freq}
\vspace{-0.15cm}
\end{equation}
 \noindent where $ \hat{\bbn}_\text{q}^{(\kappa)}$
is still i.i.d. with same statistics as $\bbn_\text{q}^{(\kappa)}$ iff $\boldsymbol \Sigma_{\text{q}\kappa} =  \sigma^2_{\text{q}\kappa} \; \bbI$. From the linearity of the expectation and from the matrix property $(\bbA \bbB)^\text{H} = \bbB^\text{H}  \bbA^\text{H}$,
the quantization noise covariance matrix becomes:
\vspace{-0.1cm}
\begin{small}
\begin{equation}
\begin{aligned}
&\mathbb{E}[ \hat{\boldsymbol \epsilon} \hat{\boldsymbol \epsilon}^\text{H} ]= 
\!\!\sum_{k_1, k_2=1}^K \!\!\!\!\!\phi_{k_1} \phi_{k_2}\!\!\!\! \sum_{\kappa_1=0}^{k_1-1} \sum_{\kappa_2=0}^{k_2-1} \!\!\! \boldsymbol \Lambda^{k_1-\kappa_1 }   \mathbb{E}\!\left[ \hat{\bbn}_\text{q}^{(\kappa_1)} (\hat{\bbn}_\text{q}^{(\kappa_2)})^\text{H} \right] \! (\! \boldsymbol \Lambda^{k_2-\kappa_2 })^\text{H}.
\end{aligned}
\label{part-MSE-error-frequency}
\vspace{-0.3cm}
\end{equation}
\end{small}

\noindent Given the quantization noise has independent realizations and a constant quantization stepsize $\Delta$ for all iterations, we can rewrite \eqref{part-MSE-error-frequency} as:
\vspace{-0.3cm}

\begin{small}
\begin{equation}
\begin{aligned}
&\mathbb{E}[ \hat{\boldsymbol \epsilon} \hat{\boldsymbol \epsilon}^\text{H} ] \!=\!\sum_{k=1}^K  \!\! \phi_{k}  \!\! \sum_{\kappa=0}^{k-1}  \!\!\boldsymbol \Lambda^{k-\kappa }  \boldsymbol \Sigma_{\text{q}\kappa}  (\boldsymbol \Lambda^{k-\kappa})^\text{H} = \! \sigma^2_{\text{q}} \! \sum_{k=1}^K  \!\! \phi_{k}  \!\! \sum_{\kappa=0}^{k-1}  \!\! \boldsymbol \Lambda^{k-\kappa }   (\boldsymbol \Lambda^{k-\kappa})^\text{H}. \\
\end{aligned}
\label{covariance-matrix-FixedStep}
\vspace{-0.15cm}
\end{equation}
\end{small}
\noindent Then, by substituting (\ref{covariance-matrix-FixedStep}) into the MSE expression $\hat{  \zeta}_\text{q} =  \frac{1}{N} \text{tr}(\mathbb{E}[ \hat{\boldsymbol \epsilon} \hat{\boldsymbol \epsilon}^H ])$  and using the relation between the Frobenius norm and the trace $ \|\bbA \|_F= \sqrt{\text{tr}(\bbA \bbA^H)}$, result (\ref{MSE-error-freq-final}) yields.

\vspace{-0.15cm}
\subsection{Proof of Corollary \ref{bound-MSE-error-freq-corollary}} \label{Appendix-bound-MSE-error-freq-corollary}
From \eqref{MSE-error-freq-final} and the relation between the $l_2$-norm and the Frobenius norm $\| \bbA \|_F\leq \sqrt{r} \| \bbA \|_2$ with $r$ the rank of $\bbA$ (at most $N$), $\hat{  \zeta}_\text{q}$ can be upper bounded as:
\begin{equation}
\begin{aligned}
  \hat{  \zeta}_\text{q}  &\leq N \frac{\sigma^2_{\text{q}}}{N}  \sum_{k=1}^K  \phi_{k}  \sum_{\kappa=0}^{k-1}  \| \boldsymbol \Lambda^{k-\kappa } \|_2^2    
    \leq  \sigma^2_{\text{q}} \sum_{k=1}^K  \phi_{k}  \sum_{\kappa=0}^{k-1}  ( \lambda_{\text{max}}^2)^{k-\kappa}. 
   \end{aligned}
   \vspace{-0.15cm}
\end{equation}

\noindent Then, with the index change $\sum_{\tau = 0}^{k-1} a^{k-\tau} =\sum_{\tau = 1}^{k} a^{\tau}$, we obtain the finite geometric series whose argument is different from one by hypothesis; thus, $\hat{  \zeta}_\text{q}$ can be upper bound as in (\ref{bound-MSE-error-freq}).
 Similarly, by  exploiting again the relationship between the $l_2$-norm and Frobenius norm of matrices ($\| \bbA \|_2 \leq  \| \bbA \|_F$) in \eqref{MSE-error-freq-final},  $\hat{  \zeta}_\text{q}$ 
can be likewise lower bounded to obtain (\ref{bound-MSE-error-freq}).

\vspace{-0.15cm}
\subsection{Proof of Proposition \ref{MSE-epsilon-DynCell}} \label{Appendix-MSE-epsilon-DynCell}
By  equivalence to (\ref{MSE-error-freq-final}),  the MSE on the filter output due to the quantization noise has
the form: 
\vspace{-0.15cm}
\begin{equation}
\begin{aligned}
 &\hat{  \zeta}_\text{q}  =\! \frac{1}{N} \! \sum_{k=1}^K  \! \phi_{k}  \! \sum_{\kappa=0}^{k-1} \! \sigma^2_{\text{q} \kappa} \| \boldsymbol \Lambda^{k-\kappa } \|_F^2    
  \leq \! \sum_{k=1}^K \! \phi_{k}  \sum_{\kappa=0}^{k-1} \!\sigma^2_{\text{q} \kappa}   ( \lambda_{\text{max}}^2)^{k-\kappa}.    \\
 \end{aligned}
 \label{MSE-error-freq-DynCell}
 \vspace{-0.15cm}
\end{equation}
By choosing the quantization stepsize $\Delta_{k}=(\lambda_{\text{max}})^{-k} \Delta_{0}$, we have:
\vspace{-0.15cm}
\begin{equation}
\begin{aligned}
 \hat{  \zeta}_\text{q}   &\leq  \frac{\Delta_{0}^2}{12} \sum_{k=1}^K  \phi_{k}  \sum_{\kappa=0}^{k-1}  ( \lambda_{\text{max}}^{-2})^{\kappa}.     
 \label{proof-bound-MSE-error-freq-DynCell}
 \end{aligned}
 \vspace{-0.15cm}
\end{equation}
The bound \eqref{proof-bound-MSE-error-freq-DynCell} contains geometric series; thus, under the assumption $\lambda_{\text{max}}>1$, we have: 
\vspace{-0.15cm}
\begin{equation}
\begin{aligned}
 \hat{  \zeta}_\text{q}  & \leq  \frac{\Delta_{0}^2}{12} \sum_{k=1}^K  \phi_{k}  \frac{1- ( \lambda_{\text{max}}^{-2})^{k}  }{1-( \lambda_{\text{max}}^{-2}) }      
   \leq  \frac{\Delta_{0}^2}{12} \sum_{k=1}^K  \phi_{k}  \frac{1}{1-( \lambda_{\text{max}}^{-2}) }    \\   
 \end{aligned}
 \vspace{-0.15cm}
\end{equation}
\noindent where the final bound can be written as \eqref{DynCell-FIR}.

\vspace{-0.15cm}
\subsection{Proof of Proposition \ref{Proposition-MSE-ARMA}} \label{Appendix-Proposition-MSE-ARMA}
By using (\ref{MSE-ARMA-QQ}), the trace cyclic property $\text{tr}(\bold A \bold B \bold C)=\text{tr}(\bold C \bold A \bold B )$, the inequality $\text{tr}(\bold A \bold B)\leq  \| \bold A\|_2 \text{tr}(\bold B)$
--which holds for any positive semi-definite matrix $\bold B \succeq \mathbf{0} $ and square matrix $\bold A$ of appropriate dimensions \cite{SaniukRhodes87}--, and   the linearity of the expectation w.r.t the trace, we can write:
\vspace{-0.15cm}
\begin{equation}
\begin{aligned}
{ \zeta}_{y t}^\text{q}&=\frac{1}{N} \mathbb{E}[ \text{tr}( (\boldsymbol 1^\top  \otimes \bbI_N)^\text{H}) (\boldsymbol 1^\top  \otimes \bbI_N) {\boldsymbol \epsilon}_{t}^\text{q} ({\boldsymbol \epsilon}_{t}^\text{q})^\text{H} ]\\
&\leq \frac{1}{N}  \|(\boldsymbol 1^\top  \otimes \bbI_N)^\text{H} (\boldsymbol 1^\top  \otimes \bbI_N) \|_2 \; \text{tr}( \mathbb{E}[  {\boldsymbol \epsilon}_{t}^\text{q} ({\boldsymbol \epsilon}_{t}^\text{q})^\text{H} ]).
\end{aligned}
\label{Proof-MSE-ARMA-quantiz}
\vspace{-0.15cm}
\end{equation}

\noindent  Then, by substituting $\boldsymbol \epsilon_t^\text{q} = \sum_{\tau = 0}^{t-1} (\boldsymbol \Psi  \otimes \bbS)^{t-\tau} \bbn_{\tau}^\text{q}$ in (\ref{Proof-MSE-ARMA-quantiz}), 
$\mathbb{E}[  {\bbn}_{\tau}^\text{q} ({\bbn}_{\tau}^\text{q})^\text{H} ]= \sigma^2_{\text{q}} \bbI$ which holds for fixed quantization stepsize in each iteration, 
and since $\|(\boldsymbol 1^\top  \otimes \bbI_N)^\text{H} (\boldsymbol 1^\top  \otimes \bbI_N) \|_2 =K$,  we can write:
  \vspace{-0.15cm}
  \begin{equation}
\begin{aligned}
{ \zeta}_{y t}^\text{q}  &\leq  \frac{K \sigma^2_{\text{q}}}{N}   \; \sum_{\tau = 0}^{t-1} \text{tr}\left((\boldsymbol \Psi  \otimes \bbS)^{t-\tau} ((\boldsymbol \Psi  \otimes \bbS)^{t-\tau})^\text{H}\right).
\end{aligned}
\label{Proof-MSE-ARMA-quantiz1}
\vspace{-0.15cm}
\end{equation}
  \noindent  By using  in (\ref{Proof-MSE-ARMA-quantiz1}) the index change $\sum_{\tau = 0}^{t-1} \bbA^{t-\tau} (\bbA^{t-\tau})^\text{H}=\sum_{\tau = 1}^{t} \bbA^{\tau} (\bbA^{\tau})^\text{H}$,
 the Frobenius norm $\| \bbA \|_F=\sqrt{\text{tr}(\bbA  \bbA^\text{H} })$, 
 the inequality $\| \bbA \|_F\leq \sqrt{r} \| \bbA \|_2$, with $r$ the rank of $\bbA$ (at most $N$), 
 and the triangle inequality of the norms $\| \bbA^2 \|_2 \leq \| \bbA \|_2^2 $,
 we have:
 \vspace{-0.15cm}
  \begin{equation}
\begin{aligned}
{ \zeta}_{y t}^\text{q} &\leq  \frac{K \sigma^2_{\text{q}}}{N}   \; \sum_{\tau = 1}^{t} \|(\boldsymbol \Psi  \otimes \bbS)^{\tau} \|_F^2
\leq  K \sigma^2_{\text{q}}  \; \sum_{\tau = 1}^{t} \|(\boldsymbol \Psi  \otimes \bbS) \|_2^{2 \tau}.
\end{aligned}
\label{Proof-MSE-ARMA-quantiz2}
\vspace{-0.15cm}
\end{equation}
Then, from the Kronecker product identity $\| \bbA \otimes \bbB \|_2 =\| \bbA \|_2  \| \bbB \|_2  $ and
   the $l_2$-norm matrix norm expression $\| \bbA \|_2=\sqrt{\text{max} \;\text{eig}(\bbA^\text{H} \bbA)}$,
 we can further rewrite \eqref{Proof-MSE-ARMA-quantiz2} as:
 \vspace{-0.15cm}
 \begin{equation}
\begin{aligned}
{ \zeta}_{y t}^\text{q} &\leq K \sigma^2_{\text{q}} \; \sum_{\tau = 1}^{t} \|\boldsymbol \Psi \|_2^{2\tau}  \| \bbS \|_2^{2\tau}
\leq  K \sigma^2_{\text{q}}   \; \sum_{\tau = 1}^{t} ( \psi_\text{max}  \lambda_{\text{max}})^{2\tau}.
\end{aligned}
\label{Proof-MSE-ARMA-quantiz33}
\vspace{-0.15cm}
\end{equation}
Finally, since \eqref{Proof-MSE-ARMA-quantiz33} is a finite geometric series with an argument smaller than one, the
quantization MSE ${ \zeta}_{y t}^\text{q} $ can be upper bounded by (\ref{upperbound-MSE-error-ARMA}).

\vspace{-0.15cm}
\subsection{Proof of Theorem \ref{Prop-MSE-epsilon-DynCell-ARMA}} \label{Appendix-Prop-MSE-epsilon-DynCell-ARMA}
By  equivalence to (\ref{Proof-MSE-ARMA-quantiz1}), but with a dynamic quantization stepsize,
  the MSE on the filter output due to the quantization noise is
upper bounded by: 
\vspace{-0.15cm}
\begin{equation}
\begin{aligned}
 &{ \zeta}_{y t}^\text{q} \leq  \frac{K }{N}   \; \sum_{\tau = 0}^{t-1} \sigma^2_{\text{q} \tau} \text{tr}((\boldsymbol \Psi  \otimes \bbS)^{t-\tau} ((\boldsymbol \Psi  \otimes \bbS)^{t-\tau})^\text{H})\\
&\leq  \! \frac{K }{12 N}   \! \sum_{\tau = 1}^{t} \! \Delta_{t-\tau}^2 \|(\boldsymbol \Psi  \otimes \bbS)^{\tau} \|_F^2
\leq  \! \frac{K }{12}   \! \sum_{\tau = 1}^{t}  \! \Delta^2_{t-\tau}  ( \psi_\text{max}  \lambda_{\text{max}})^{2\tau}\\
\end{aligned}
 \label{Proof-MSE-error-DynCell-ARMA}
 \vspace{-0.15cm}
\end{equation}
 \noindent  where similarily to (\ref{Proof-MSE-ARMA-quantiz2}) and (\ref{Proof-MSE-ARMA-quantiz33}),
 we changed the summatiom index, used the expression of the Frobenius norm $\| \bbA \|_F=\sqrt{\text{tr}(\bbA  \bbA^\text{H} })$, and leveraged
the norm properties.
 
For the quantization stepsize $\Delta_{\tau}=(\psi_\text{max} \lambda_{\text{max}})^{\tau} \Delta_{0}$, \eqref{Proof-MSE-error-DynCell-ARMA} can be further upper bounded as:
 \vspace{-0.15cm}
 \begin{equation}
\begin{aligned}
{ \zeta}_{y t}^\text{q}  &\leq \frac{ K}{12} \sum_{\tau = 1}^{t} (\psi_\text{max}  \lambda_{\text{max}})^{2t}  \Delta_{0}    
  \end{aligned}
  \vspace{-0.15cm}
 \end{equation}
which can be easily rephrased as in (\ref{MSE-epsilon-DynCell-ARMA}).

\vspace{-0.15cm}
\subsection{Quantized FIR graph filter  over time-varying graphs}\label{Appendix-FIR-output-varying-graph}
Considering  ${\bold x}^{(0)} =\bold x $ and the  quantized message at iteration $k$,
$\tilde{\bold x}^{(k)}= {\bold x}^{(k)}+ \bbn_\text{q}^{(k)}$, the output of the shifted signal with quantization 
performed over $\mathcal{G}_{t}$ is:
\vspace{-0.05cm}
\begin{small}
\begingroup
\allowdisplaybreaks
\begin{align}
  {\bold x}^{(1)} &= \bold S_{t{-}1} \tilde{\bold x}^{(0)}  = \bold S_{t{-}1}  ({\bold x}^{(0)}{+} \bbn_\text{q}^{(0)})= \bold S_{t{-}1}  {\bold x}^{(0)}{+} \bold S_{t{-}1} \bbn_\text{q}^{(0)} \nonumber\\
  {\bold x}^{(2)} &= \bold S_{t{-}2} \tilde{\bold x}^{(1)} 
  =  \bold S_{t{-}2} \bold S_{t{-}1} {\bold x}^{(0)}{+}   \bold S_{t{-}2} \bold S_{t{-}1} \bbn_\text{q}^{(0)}{+}  \bold S_{t{-}2} \bbn_\text{q}^{(1)} \nonumber\\
  \vdots  & \nonumber  \\
 {\bold x}^{(k)} &= \; \bigg(\prod_{\tau=t-1}^{t-k}  \bold S_{\tau}  \bigg) \; \bold x^{(0)} + \sum_{\kappa=0}^{k-1}  \; \bigg(\prod_{\tau=t-1-\kappa}^{t-k}  \bold S_{\tau}  \bigg) \; \bbn_\text{q}^{(\kappa)}, \; \; k \geq 1
  \end{align}
  \endgroup
\end{small} 
\vspace{-0.25cm}

\noindent The quantized output of  FIR graph filter at iteration $t$, performed over $\mathcal{G}_{t}$ with quantization effects, is given by:
\vspace{-0.1cm}
\begin{small}
\begingroup
\allowdisplaybreaks
\begin{align}
 &{\bold y}_{t}^\text{q} =  \phi_0  \; \bold x +    \sum_{k=1}^{K}  \phi_k \; {\bold x}^{(k)} \nonumber \\
      & = \phi_0  \; \bold x +    \sum_{k=1}^{K}  \phi_k  \bigg(  \bold \Theta(t{-}1,t{-}k) \;\bold x^{(0)} +    \sum_{\kappa=0}^{k-1}   \bold \Theta(t{-}1{-}\kappa,t{-}k)\;    \bbn_\text{q}^{(\kappa)} \bigg) \nonumber\\
      &= \sum_{k=0}^{K} \! \phi_k \bold \Theta(t{-}1,t{-}k)  \bold x +  \sum_{k=1}^{K} \sum_{\kappa=0}^{k-1}   \phi_k \; \bold \Theta(t{-}1{-}\kappa,t{-}k)  \; \bbn_\text{q}^{(\kappa)}. 
\vspace{-0.25cm}
\end{align}
\endgroup
 \end{small}

 \vspace{-0.35cm}
\subsection{Proof of Proposition \ref{Prop-MSE-Vary-graph}} \label{Appendix-Prop-MSE-Vary-graph}
  By using  $\|\bold x \|_2^2=\text{tr} (\bold x \bold x^\text{H})$
 and  rearranging the   summation indices, we can write  the  MSE of the filter output due quantization and graph randomness as:
 \vspace{-0.15cm}
 \begin{equation}
\begin{aligned}
  \zeta^\text{q}_t&=  \mathbb{E}[\frac{1}{N} \text{tr}({\boldsymbol \epsilon_t} {\boldsymbol \epsilon_t}^H )]= \frac{1}{N} \mathbb{E}[\|\boldsymbol \epsilon_t\|_2^2]  \\
      &  =\frac{1}{N} \mathbb{E} \bigg[ \bigg \| \sum_{\kappa=1}^{K}  \sum_{k=\kappa}^{K}  \phi_k  \; \bold \Theta(t{-}\kappa,t{-}k)    \; \bbn_\text{q}^{(\kappa-1)} \bigg \|_2^2 \bigg].\\
   \end{aligned}\label{rearrangingMSE-Vary-graph}
   \vspace{-0.15cm}
 \end{equation}

\noindent Let then vector $\boldsymbol \omega(\kappa,t)=\sum_{k=\kappa}^{K}  \phi_k  \bold \Theta(t{-}\kappa,t{-}k)    \;   \bbn_\text{q}^{(\kappa-1)}$ account for the accumulated quantization noise over time-varying graphs. By using  $\|\bold x \|_2^2=\bold x^\text{H} \bold x $, we can write: 
 \vspace{-0.25cm}
 
  \begin{small}
  \begin{equation}
\begin{aligned}
  \mathbb{E} \bigg[\bigg\| \sum_{\kappa=1}^{K} \boldsymbol \omega(\kappa,t) \bigg\|_2^2 \bigg] 
   = \sum_{\kappa_1=1}^{K} \sum_{\kappa_2=1}^{K} \mathbb{E} \bigg[ \boldsymbol \omega(\kappa_1,t)^\text{H}  \boldsymbol \omega(\kappa_2,t) \bigg].\\
  \end{aligned}
   \vspace{-0.25cm}
 \end{equation}
  \end{small}
  
 \noindent Since the quantization errors are zero mean  
 and independent from  graph topology processes, we have:
 \vspace{-0.25cm}
 
   \begin{small}
 \begin{equation}
\mathbb{E} \bigg[ \boldsymbol \omega(\kappa_1,t)^\text{H} \boldsymbol \omega(\kappa_2,t) \bigg]= \left\{
    \begin{array}{ll}
        0 & \text{if} \; \kappa_1\neq \kappa_2\\
        \mathbb{E} [ \|\boldsymbol \omega(\kappa_1,t)\|_2^2]  & \text{if} \; \kappa_1= \kappa_2.
    \end{array}
\right.
  \label{equalitySimplification}
   \vspace{-0.25cm}
\end{equation}
  \end{small}

\noindent Therefore, we can rewrite \eqref{equalitySimplification} as:
 \vspace{-0.25cm}
 
  \begin{small}
   \begin{equation}
\begin{aligned}
  \mathbb{E} \bigg[\bigg\| \sum_{\kappa=1}^{K} \boldsymbol \omega(\kappa,t) \bigg\|_2^2 \bigg]& =
   \sum_{\kappa=1}^{K}  \mathbb{E} \bigg[ \big \|\boldsymbol \omega(\kappa,t) \big\|_2^2 \bigg].\\
  \end{aligned}
  \label{equalitySomme}
   \vspace{-0.25cm}
 \end{equation}
   \end{small}

\noindent  Using once again the norm property $\|\bold x \|_2^2=\text{tr} (\bold x \bold x^\text{H})$, the cyclic property of the trace $\text{tr}(\bold A \bold B \bold C)=\text{tr}(\bold C \bold A \bold B )$,
 and  the commutativity
 of the trace to respect to the expectation, we can write:
  \vspace{-0.25cm}
 
 \begin{small}
  \begingroup
\allowdisplaybreaks
\begin{align}
  & \mathbb{E} [ \| {\boldsymbol \epsilon_t}  \|_2^2  ] = \sum_{\kappa=1}^{K} \mathbb{E} \bigg[ \text{tr} \bigg( \bigg(\sum_{k=\kappa}^{K}  \phi_k  \bold \Theta(t{-}\kappa,t{-}k)     \bigg) \; \bbn_\text{q}^{(\kappa-1)} ({\bbn_\text{q}^{(\kappa-1)}})^\text{H} \nonumber\\
 &\qquad\qquad   \times \; \bigg(\sum_{k=\kappa}^{K} \phi_k  {\bold \Theta(t{-}\kappa,t{-}k)}       \bigg)^\text{H} \bigg ) \bigg] \nonumber\\
   &    = \sum_{\kappa=1}^{K} \text{tr} \bigg(  \mathbb{E} \bigg[  \bigg(\sum_{k=\kappa}^{K}  \phi_k \bold \Theta(t{-}\kappa,t{-}k)^\text{H}       \bigg)   \bigg(\sum_{k=\kappa}^{K}  \phi_k   \bold \Theta(t{-}\kappa,t{-}k) \bigg) \bigg] \nonumber\\
  &\qquad\qquad  \times    \mathbb{E} \big[ \bbn_\text{q}^{(\kappa-1)} ({\bbn_\text{q}^{(\kappa-1)}})^\text{H}   \big] \bigg ).
 \end{align}
 \endgroup
  \vspace{-0.45cm}
 \end{small}

\noindent By using the  inequality $\text{tr}(\bold A \bold B)\leq  \| \bold A\|_2 \;\text{tr}(\bold B)$, we obtain:
  \vspace{-0.35cm}
  
  \begin{small}
   \begin{equation}
\begin{split}
   & \mathbb{E} [ \| {\boldsymbol \epsilon_t}  \|_2^2  ]  \leq  \sum_{\kappa=1}^{K} \text{tr} \bigg(  \mathbb{E} \big[ \bbn_\text{q}^{(\kappa-1)} ({\bbn_\text{q}^{(\kappa-1)}})^\text{H}   \big] \bigg )\\
   &  \times \bigg \|   \mathbb{E} \bigg[ \bigg(\sum_{k=\kappa}^{K}  \phi_k  \bold \Theta(t{-}\kappa,t{-}k)^\text{H}       \bigg)   \bigg(\sum_{k=\kappa}^{K}  \phi_k \bold \Theta(t{-}\kappa,t{-}k)     \bigg) \bigg] \bigg \|_2. \\
      \end{split}
       \vspace{-0.35cm}
 \end{equation}
 \end{small}
 \vspace{-0.25cm}

 \noindent Since $\text{tr} \big(  \mathbb{E} \big[ \bbn_\text{q}^{(\kappa)} ({\bbn_\text{q}^{(\kappa)}})^\text{H}   \big] \big )= \text{tr} \big( \sigma_{\text{q} \kappa}^2  \bold I \big)= N \sigma_{\text{q} \kappa}^2$
 and
using the Jensen's inequality of the spectral norm ($\| \mathbb{E}[\bold A] \|_2 \leq \mathbb{E}[ \| \bold A \|_2  ]$),
  we can further write:
  \vspace{-0.15cm}
  \begin{small}
    \begin{equation}
\begin{split}
   & \mathbb{E} [ \| {\boldsymbol \epsilon_t}  \|_2^2  ]  \leq  \\ 
   & N \!\sum_{\kappa=1}^{K} \!\! \sigma_{\text{q} \kappa-1}^2  \mathbb{E} \! \left[ \bigg \| \!  \bigg( \sum_{k=\kappa}^{K} \!\! \phi_k  {\bold \Theta(t{-}\kappa,t{-}k)}^\text{H}      \bigg)  \!\!  \bigg(  \! \sum_{k=\kappa}^{K} \!  \phi_k  \bold \Theta(t{-}\kappa,t{-}k) \!  \bigg) \bigg \|_2  \! \right]  \\     
       &\leq   \frac{N}{12} \sum_{\kappa=1}^{K}  \Delta_{\kappa{-}1}^2 \; \mathbb{E} \big[\Upsilon(t,\kappa) \big]\\
\end{split}
 \label{mideqresult}
 \end{equation}
 \end{small}
 \vspace{-0.15cm}
  \noindent  where $ \Upsilon(t,\kappa)$ is:
  \vspace{-0.15cm}
    \begin{equation}\label{eq.proofUpsilon}
 \Upsilon(t,\kappa)=\bigg \| \!  \bigg( \sum_{k=\kappa}^{K} \!\! \phi_k  {\bold \Theta(t{-}\kappa,t{-}k)}^\text{H}      \bigg)  \!\!  \bigg(  \! \sum_{k=\kappa}^{K}  \phi_k  \bold \Theta(t{-}\kappa,t{-}k) \!  \bigg) \bigg \|_2.
  \vspace{-0.15cm}
  \end{equation}

 \noindent  By using the spectral norm sub-multiplicativity $\| \bold A \bold B  \|_2  \leq \| \bold A   \|_2  \|  \bold B  \|_2 $ and subadditivity $\| \bold A +\bold B  \|_2  \leq \| \bold A   \|_2 + \|  \bold B  \|_2 $ along with the upper bound of the shift operator $\| \bold S_t  \|_2  {\leq} \| \bold S  \|_2  {\leq} \rho$ for all $t$, we upper bound \eqref{eq.proofUpsilon} as:
\vspace{-0.25cm}
 
  \begin{small}
    \begin{equation}
\begin{split}
    \Upsilon(t,\kappa) & \leq       \bigg \| \!   \sum_{k=\kappa}^{K} \! \phi_k {\bold \Theta(t{-}\kappa,t{-}k)}^\text{H}     \bigg \|_2  \;\bigg \|  \! \sum_{k=\kappa}^{K}   \phi_k  \bold \Theta(t{-}\kappa,t{-}k)  \bigg \|_2    \\     
    &\leq     \bigg ( \!   \sum_{k=\kappa}^{K}  \rho^{k-\kappa+1} \;  | \phi_k |  \bigg )^2.       
 \end{split}
 \label{endeqresult}
 \vspace{-0.25cm}
 \end{equation}
 \end{small}

Finally, by substituting \eqref{endeqresult} into \eqref{mideqresult} and computing the expectation, $\zeta^\text{q}_t$ can be upper bounded by (\ref{eq_prop18bound}).

\vspace{-0.15cm}
\subsection{Proof of Theorem \ref{Prop-MSE-epsilon-DynCell-ARMA-TVGraph}} \label{Appendix-Prop-MSE-epsilon-DynCell-ARMA-TVGraph}
\noindent Similarly to \eqref{Proof-MSE-ARMA-quantiz}, we can write the MSE of ARMA filter due to quantization and graph randomness \eqref{MSE-ARMA-time-varying-QQ} as:
\vspace{-0.25cm}

\begin{small}
\begin{equation}
\begin{aligned}
{ \xi}_{y t}^\text{q} 
&\leq \! \frac{1}{N} \! \|(\boldsymbol 1^\top \! \otimes \bbI_N)^\text{H} (\boldsymbol 1^\top \! \otimes \bbI_N) \|_2  \text{tr}( \mathbb{E}[  {\boldsymbol \varepsilon}_{t}^\text{q} ({\boldsymbol \varepsilon}_{t}^\text{q})^\text{H} ])
\leq \! \frac{K}{N}   \text{tr}( \mathbb{E}[  {\boldsymbol \varepsilon}_{t}^\text{q} ({\boldsymbol \varepsilon}_{t}^\text{q})^\text{H} ])
\end{aligned}
\label{Proof-MSE-ARMA-quantiz-time-varying}
\end{equation}
\end{small}
\vspace{-0.25cm}

\noindent Then, by substituting $\boldsymbol \varepsilon_t^\text{q}$ with its expression,
using  the linearity of the expectation w.r.t the trace, the  cyclic property of the trace $\text{tr}(\bold A \bold B \bold C)=\text{tr}(\bold C \bold A \bold B )$, we can write:  
 \vspace{-0.25cm}
 
 \begin{small}
 \begin{equation}
\begin{aligned}
  & \text{tr}( \mathbb{E}[  {\boldsymbol \varepsilon}_{t}^\text{q} ({\boldsymbol \varepsilon}_{t}^\text{q})^\text{H} ]) =  \sum_{\tau_1=0}^{t-1} \sum_{\tau_2=0}^{t-1} \mathbb{E} \bigg[ \text{tr} \bigg( \big(\prod_{\varsigma=\tau_2}^{t-1} \boldsymbol \Psi  \otimes \bbS_{\varsigma} \big)^\text{H}  \\ 
  & \times \big(\prod_{\varsigma=\tau_1}^{t-1}  \boldsymbol \Psi  \otimes \bbS_{\varsigma} \big) \; \bbn_{\tau_1}^\text{q} (\bbn_{\tau_2}^\text{q})^\text{H}   \bigg) \bigg]\\
  &=  \sum_{\tau_1=0}^{t-1} \sum_{\tau_2=0}^{t-1}  \text{tr} \bigg(  \mathbb{E} \bigg[\big(\prod_{\varsigma=\tau_2}^{t-1} \boldsymbol \Psi  \otimes \bbS_{\varsigma} \big)^\text{H}  \big(\prod_{\varsigma=\tau_1}^{t-1}  \boldsymbol \Psi  \otimes \bbS_{\varsigma} \big) \bigg] \\ 
  &  \times  \mathbb{E} [ \bbn_{\tau_1}^\text{q} (\bbn_{\tau_2}^\text{q})^\text{H}]   \bigg)
\end{aligned}
\end{equation}
\end{small}
\vspace{-0.2cm}

\noindent By  considering
$\mathbb{E} [ \bbn_{\tau_1}^\text{q} (\bbn_{\tau_2}^\text{q})^\text{H}]=\bold 0 $ if $\tau_1 \neq \tau_2$,
using the  inequality $\text{tr}(\bold A \bold B)\leq  \| \bold A\|_2 \;\text{tr}(\bold B)$,
assuming unified quantization with dynamic stepsize i.e., $\text{tr} (\mathbb{E} [ \bbn_{\tau}^\text{q} (\bbn_{\tau}^\text{q})^\text{H}])=K N  \sigma_{\text{q},\tau}^2 $
and using the Jensen's inequality of the spectral norm ($\| \mathbb{E}[\bold A] \|_2 \leq \mathbb{E}[ \| \bold A \|_2  ]$), we can write:
 \vspace{-0.18cm}
 \begin{small}
 \begin{equation}
\begin{aligned}
   \text{tr}( \mathbb{E}[  {\boldsymbol \varepsilon}_{t}^\text{q} ({\boldsymbol \varepsilon}_{t}^\text{q})^\text{H} ]) 
  & \leq K N    \sum_{\tau=0}^{t-1} \! \sigma_{\text{q},\tau}^2     \mathbb{E} \bigg[ \bigg\| \big(\!\prod_{\varsigma=\tau}^{t-1} \boldsymbol \Psi  \otimes \bbS_{\varsigma} \big)^\text{H}  \big(\!\prod_{\varsigma=\tau}^{t-1}  \boldsymbol \Psi  \otimes \bbS_{\varsigma} \big) \bigg\|_2 \bigg]  \\ 
\label{Proof-MSE-ARMA-quantiz-time-varying-Part2}
  \end{aligned}
\end{equation}
 \end{small}
 \vspace{-0.25cm}

\noindent By using  the sub-multiplicativity  property of the spectral norm
 of a square matrix i.e., $\| \bold A \bold B  \|_2  \leq \| \bold A   \|_2  \|  \bold B  \|_2 $,
the property $\| \bbA \otimes \bbB\|_2 = \| \bold A\|_2 \| \bold B\|_2$ and 
 assuming that the spectral norm of the shift operator used is upper bounded
 i.e.,  $\| \bbS \|_2  \leq \| \bbS_t \|_2 \leq \rho$ for all $t$, we have:
 \vspace{-0.25cm}
 
 \begin{small}
  \begin{equation}
\begin{aligned}
 \bigg\|  \big(\!\prod_{\varsigma=\tau}^{t-1} \! \boldsymbol \Psi  \otimes \bbS_{\varsigma} \big)^\text{H}  \big(\!\prod_{\varsigma=\tau}^{t-1}  \! \boldsymbol \Psi  \otimes \bbS_{\varsigma} \big)  \bigg\|_2 & \!\! \! \leq \!\!
    \big(\!\prod_{\varsigma=\tau}^{t-1} \! \|  \boldsymbol \Psi \|_2 \! \|  \bbS_{\varsigma}  \|_2 \big)    \big(\!\prod_{\varsigma=\tau}^{t-1} \! \|  \boldsymbol \Psi  \|_2 \! \|  \bbS_{\varsigma}  \|_2 \big) \\
 &\leq (\psi_\text{max} \; \rho)^{2t-2\tau} 
 \label{Proof-MSE-ARMA-quantiz-time-varying-bound}
\end{aligned}
\vspace{-0.25cm}
\end{equation}
 \end{small}
 
\noindent  By applying the expectation to \eqref{Proof-MSE-ARMA-quantiz-time-varying-bound} and combining it with 
 \eqref{Proof-MSE-ARMA-quantiz-time-varying-Part2} and \eqref{Proof-MSE-ARMA-quantiz-time-varying}, and
 making an index change using $\sum_{\tau = 0}^{t-1} c_\tau a^{t-\tau} =\sum_{\tau = 1}^{t} c_{t-\tau} a^{\tau}$,  
we can write:
\vspace{-0.15cm}
\begin{equation}
\begin{aligned}
{ \xi}_{y t}^\text{q} &\leq K^2   \sum_{\tau=0}^{t-1}  \sigma_{\text{q},\tau}^2 \big((\psi_\text{max} \; \rho)^2 \big)^{t-\tau} \leq K^2   \sum_{\tau=1}^{t} \sigma_{\text{q},t-\tau}^2 \big((\psi_\text{max} \; \rho)^2 \big)^{\tau}\\
                      & \leq \frac{K^2}{12}   \sum_{\tau=1}^{t}  \Delta_{t-\tau}^2 (\psi_\text{max} \; \rho)^{2\tau}
\end{aligned}
\label{Proof-MSE-ARMA-quantiz-time-varying-final-dyn}
\vspace{-0.1cm}
\end{equation}

\noindent With  the choice of the quantization stepsize $\Delta_{\tau}=(  \psi_\text{max} \; \rho)^{\tau} \Delta_{0}$, 
 the final bound in \eqref{Proof-MSE-ARMA-quantiz-time-varying-final-dyn}  becomes:
 \vspace{-0.15cm}
 \begin{equation}
\begin{aligned}
{ \xi}_{y t}^\text{q} &\leq \frac{ K^2 }{12} \sum_{\tau = 1}^{t} (\psi_\text{max} \; \rho)^{2t}  \Delta_{0}    
  \end{aligned}
  \vspace{-0.15cm}
 \end{equation}
\noindent Therefore, ${\xi}_{y t}^\text{q}$ can be upper bounded by \eqref{MSE-epsilon-DynCell-ARMA-TVGraph}.

\vspace{-0.1cm}
\bibliographystyle{IEEEtran}
\bibliography{references,dissertation}

\end{document}